\documentclass[a4paper,UKenglish,cleveref,thm-restate,numberwithinsect]{lipics-v2021}
\pdfoutput=1
\usepackage{microtype}
\usepackage{relsize}

\usepackage[paper]{knowledge}
\knowledgeconfigure{notion}

\input{globals/knowledges.kl}

\usepackage{import}
\usepackage{pdfpages}
\usepackage{transparent}
\usepackage{xcolor}

\usepackage[ruled]{algorithm2e}

\usepackage{pdfcomment}

\usepackage{amssymb}
\usepackage{booktabs}
\usepackage{stmaryrd}
\usepackage{faktor}

\usepackage{upgreek}

\usepackage{tikz}
\usetikzlibrary{positioning}
\usepackage{tikz-cd}

\usepackage{paralist}
\usepackage{bussproofs}

\usepackage{environ}
\NewEnviron{killcontents}{}
\newif\ifsubmission
\newif\iflipics
\newif\iflncs
\newenvironment{sendappendix}{
	\ifsubmission \killcontents \fi
}{
	\ifsubmission \endkillcontents \fi
}
\submissiontrue
\newcommand\appref[1]{the full paper}

\newcommand{\kvdef}[2]{%
	\expandafter\newrobustcmd\csname #1\endcsname{\kl[{math.data.#1}]{#2}}%
	\knowledge{math.data.#1}{notion}%
}

\newcommand{\reldef}[2]{%
	\expandafter\newrobustcmd\csname #1\endcsname{\mathrel{\kl[{math.data.#1}]{#2}}}%
	\knowledge{math.data.#1}{notion}%
}

\newcommand\defined{:=}

\reldef{orth}{\bot}

\reldef{HigLeq}{\leq_*}
\reldef{KruLeq}{\leq_\mathsf{tree}}

\kvdef{TrivTop}{\tau_{\text{triv}}}
\kvdef{AExp}{\mathsf{E}}
\kvdef{BadIterator}{\mathsf{R_{\textsf{\textnormal{pref}}}}}
\kvdef{RegSubExp}{\mathsf{E_{\textnormal{words}}}}
\kvdef{DivExp}{\mathsf{E_{\Diamond}}}
\kvdef{RegSubExpOm}{\mathsf{E_{\operatorname{\omega-words}}}}
\kvdef{RegSubExpOrd}{\mathsf{E_{\operatorname{\alpha-words}}}}
\kvdef{TreeExp}{\mathsf{E_{\operatorname{tree}}}}
\kvdef{InfRegSubExp}{\mathsf{E_{\operatorname{inf. words}}}}
\kvdef{PieceExp}{\mathsf{E_{\operatorname{piece}}}}
\kvdef{TreeExpOm}{\mathsf{E_{\operatorname{\alpha-trees}}}}

\newcommand{\kdef}{\knowledgenewrobustcmd}

\kdef\Alex[1]{\cmdkl{\mathsf{alex}(}{#1}\cmdkl{)}}
\kdef\omtrees[2]{\cmdkl{\mathsf{T}^{<\alpha}(}{#1}\cmdkl{)}}
\kdef\trees[1]{\cmdkl{\mathsf{T}(}{#1}\cmdkl{)}}
\kdef\words[1]{{#1}^{\cmdkl{*}}}
\kdef\fSets[1]{\cmdkl{\mathsf{P_f}(}#1\cmdkl{)}}
\kdef\pSets[1]{\cmdkl{\mathsf{P}(}#1\cmdkl{)}}
\kdef\fmSets[1]{{#1}^{\cmdkl{\circledast}}}
\kdef\openword[1]{\cmdkl{[}{#1}\cmdkl{]}}
\kdef\spec[1]{\mathrel{\cmdkl{\leq}_{#1}}}
\kdef\upset[2]{\cmdkl{\uparrow}_{#1}{#2}}
\kdef\lfp[2]{\cmdkl{\mathsf{lfp}_{#1}.}{#2}}
\kdef\id{\cmdkl{\mathsf{id}}}
\kdef\singleton{\cmdkl{\mathbf{1}}}
\newcommand\Nat{\mathbb{N}}
\kdef\PrefTopo{\cmdkl{\tau_{\textnormal{pref}^*}}}

\newrobustcmd\DExpand[2]{\kl(noeth)[\DExpand]{{F}_{#1,#2}}}
\knowledge{\DExpand}{notion,scope=noeth}

\newrobustcmd\splitW{\kl(noeth)[\splitW]{\operatorname{split}}}
\knowledge{\splitW}{notion,scope=noeth}

\newrobustcmd\supportW{\kl(noeth)[\supportW]{\operatorname{supp}}}
\knowledge{\supportW}{notion,scope=noeth}

\newrobustcmd\CatOrd{\kl[\CatOrd]{\mathsf{Ord}}}
\knowledge\CatOrd{notion}

\newrobustcmd\CatPO{\kl[\CatPO]{\mathsf{PO}}}
\knowledge\CatPO{notion}

\newrobustcmd\CatEl{\kl[\CatEl]{\operatorname{el}}}
\knowledge\CatEl{notion}

\newrobustcmd\CatSet{\kl[\CatSet]{\mathsf{Set}}}
\knowledge\CatSet{notion}

\newrobustcmd\CatTop{\kl[\CatTop]{\mathsf{Top}}}
\knowledge\CatTop{notion}

\newrobustcmd\supp{\kl[\supp]{\operatorname{supp}}}
\knowledge\supp{notion}

\newrobustcmd\letters{\kl[\letters]{{\operatorname{letters}}}}
\knowledge\letters{notion}

\newrobustcmd\restr[2]{\kl[\restr]{{#1}|{#2}}}
\knowledge\restr{notion}

\newrobustcmd\Hom{\kl[\Hom]{\operatorname{Hom}}}
\knowledge\Hom{notion}

\newrobustcmd\Aut{\kl[\Aut]{\operatorname{Aut}}}
\knowledge\Aut{notion}

\newrobustcmd\slice[2]{\kl[\slice]{{#1}/{#2}}}
\knowledge\slice{notion}

\newrobustcmd\Image{\kl[\Image]{\operatorname{Im}}}
\knowledge\Image{notion}

\newrobustcmd\ACat{\kl[\ACat]{\mathcal{C}}}
\knowledge\ACat{notion}

\newrobustcmd\AFunc{\kl[\AFunc]{\mathsf{G}}}
\knowledge\AFunc{notion}

\newrobustcmd\ARef{\kl[\ARef]{\mathsf{R}}}
\knowledge\ARef{notion}

\newrobustcmd\Id{\kl[\Id]{\operatorname{Id}}}
\knowledge\Id{notion}

\newrobustcmd\compose{\mathrel{\kl[\compose]{\circ}}}
\knowledge\compose{notion}

\newrobustcmd{\osub}{\mathrel{\kl[\osub]{{\triangleleft}}}}
\newrobustcmd{\osubeq}{\mathrel{\kl[\osubeq]{{\trianglelefteq}}}}
\newrobustcmd{\depth}{\kl[\depth]{\operatorname{\mathsf{depth}}}}
\knowledge\osub{notion}
\knowledge\depth{notion}
\knowledge\osubeq{notion}

\newrobustcmd{\sstruct}{\mathrel{\kl[\sstruct]{{\sqsubset}}}}
\newrobustcmd{\sstructeq}{\mathrel{\kl[\sstructeq]{{\sqsubseteq}}}}
\knowledge\sstruct{notion}
\knowledge\sstructeq{notion}

\newrobustcmd{\dl}{\mathrel{\kl[\dl]{{\triangleleft}}}}
\newrobustcmd{\dleq}{\mathrel{\kl[\dleq]{{\trianglelefteq}}}}
\knowledge\dl{notion}
\knowledge\dleq{notion}

\newrobustcmd\AUseq{\kl[\AUseq]{\mathcal{U}}}
\knowledge\AUseq{notion}

\newrobustcmd\AUseqP{\kl[\AUseqP]{\mathcal{U}'}}
\knowledge\AUseqP{notion}

\newrobustcmd\U{\kl[\U]{\mathsf{U}}}
\knowledge\U{notion}

\newrobustcmd\SCons{\kl[\SCons]{\mathsf{G}}}
\knowledge\SCons{notion}

\newrobustcmd\TCons{\kl[\TCons]{\mathsf{G}'}}
\knowledge\TCons{notion}

\newrobustcmd\OCons{\kl[\TCons]{\mathsf{G}^O}}
\knowledge\OCons{notion}

\newrobustcmd\InitSCons{\kl[\InitSCons]{\mu\SCons}}
\knowledge\InitSCons{notion}

\newrobustcmd\CatLim[1]{\kl[\CatLim]{\operatorname{lim}(#1)}}
\knowledge\CatLim{notion}

\newrobustcmd\LDiag{\kl[\LDiag]{D}}
\knowledge\LDiag{notion}

\newrobustcmd\adh[1]{\kl[\adh]{\overline{#1}}}
\knowledge\adh{notion}

\newrobustcmd\TDown{\kl[\TDown]{\operatorname{\mathsf{Down}}}}
\knowledge\TDown{notion}

\newrobustcmd\UTopo{\kl[\UTopo]{\mathcal{U}}}
\knowledge\UTopo{notion}

\kvdef{reptau}{\tau_{\operatorname{rep}}}

\newcommand{\setof}[2]{\left\{ {#1} {\colon} {#2} \right\}}
\newcommand{\set}[1]{\left\{ #1 \right\}}

\renewcommand{\P}{\mathbb{P}}

\kvdef{Graph}{\textnormal{Graphs}}

\definecolor{myBlue}{HTML}{88C0D0}
\definecolor{myDarkBlue}{HTML}{5E81AC}
\definecolor{myTeal}{HTML}{8FBCBB}
\definecolor{myOrange}{HTML}{D08770}
\definecolor{myGreen}{HTML}{A3BE8C}
\definecolor{myRed}{HTML}{BF616A}
\definecolor{myYellow}{HTML}{EBCB8B}
\definecolor{myPurple}{HTML}{B48EAD}
\definecolor{myBlack}{HTML}{3B4252}

\newcommand{\QO}{\mathcal{E}}

\newcommand{\SP}{\mathcal{X}}

\newcommand{\TP}{\ensuremath{\tau}}

\newcommand{\Alpha}{\ensuremath{\Sigma}}

\newcommand{\becomes}{\leftarrow}

\submissionfalse
\lipicstrue
\lncsfalse

\usepackage[numbers]{natbib} 
\makeatletter
\providecommand{\doi}[1]{\href{https://doi.org/#1}{\nolinkurl{doi:#1}}}

\title{Fixed Points and Noetherian Topologies}
\titlerunning{FPNT}
\author{Aliaume Lopez}{%
Université Paris Cité, CNRS, IRIF, F-75013, Paris, France \and
Université Paris-Saclay, CNRS, ENS Paris-Saclay, Laboratoire Méthodes Formelles,
91190, Gif-sur-Yvette, France.%
}{aliaume.lopez@ens-paris-saclay.fr}%
{https://orcid.org/0000-0002-4205-327X}{}
\authorrunning{A.\ Lopez}
\Copyright{Aliaume Lopez}



\relatedversion{A full version of the paper is available at
\url{https://doi.org/10.48550/arXiv.2207.07614}.}
\acknowledgements{I thank Jean Goubault-Larrecq and Sylvain Schmitz
for their help and support in writing this paper.
I thank Simon Halfon for his help on transfinite words.
}

\EventEditors{EVENT EDITORS}
\EventNoEds{2}
\EventLongTitle{CONFERENCE}
\EventShortTitle{CONFERENCE}
\EventAcronym{CONFERENCE}
\EventYear{2022}
\EventDate{Conference date}
\EventLocation{Conference location}
\EventLogo{}
\SeriesVolume{183}
\ArticleNo{29}

\nolinenumbers
\hideLIPIcs

\begin{document}

\maketitle

\begin{abstract}
	Noetherian spaces are a generalisation of well-quasi-ordering
to topologies, that can be used to prove termination of programs.
They
find applications in the verification of transitions systems, that
are better described using topology.
The goal of this paper is to
allow the systematic description of computations using inductively defined datatypes
using Noetherian spaces. This is achieved through
a fixed point theorem based on a
topological minimal bad sequence argument.

\end{abstract}


\section{Introduction}
\label{sec:intro}
The goal of this paper is to bring
inductively defined datatypes to the theory of \kl{Noetherian spaces}.
Before that, let us give some context on the history and relevance
of this concept.

\subparagraph*{Well-quasi-orderings.}
\AP
Let $(\QO, \leq)$ be a set endowed with a quasi-order.
A sequence $(x_n)_n \in \QO^{\mathbb{N}}$ is \intro*\kl(wqo){good} whenever
there exists $i < j$ such that $x_i \leq x_j$.
A quasi-ordered set $(\QO, \leq)$ is a \intro{well-quasi-ordered}
if every sequence is \kl(wqo){good}. By calling a sequence \kl(wqo){bad}
whenever it is not good, well-quasi-orderings
are equivalently defined as having no infinite \kl(wqo){bad sequences}.
This generalisation of well-founded total orderings
can be used as a basis for proving program termination.
For instance, algorithms
alike \cref{ex:intro:enum-planar} can be studied
via well-quasi-ordering and the length of their
bad sequences \cite{figueira_ackermannian_2011}.
More generally, one can map the states of a run to a \kl{wqo}
via a so-called quasi-ranking function to both prove the termination
of the program and gain information about its runtime
\cite[Chapter 2]{schmitzHdr}.

\begin{example}
	\label{ex:intro:enum-planar}
	Let $\mathsf{Alg}$ be the algorithm
	with three integer variables $a,b,c$ that
	does one of the following operations
	$l \colon \langle a,b,c \rangle
		\becomes \langle a - 1, b , 2c  \rangle$;
	$r \colon \langle a,b,c \rangle
		\becomes \langle 2c, b -1 , 1 \rangle$;
	until $a$, $b$ or $c$ becomes negative.

	For every choice of $a,b,c \in \mathbb{N}^3$,
	the algorithm $\mathsf{Alg}$
	builds
	a bad sequence of triples when ordering $\mathbb{N}^3$
	with $(a_1,b_1,c_1) \leq (a_2, b_2, c_2)$ whenever
	$a_1 \leq a_2$, $b_1 \leq b_2$, and
	$c_1 \leq c_2$.
	Because $(\Nat^3, \leq)$
	is a \kl{well-quasi-ordering}
	\cite[see Dickson's Lemma in][]{lectureNotesWqo},
	$\mathsf{Alg}$
	terminates for every choice of initial triple $(a,b,c) \in \Nat^3$.
\end{example}

As a combinatorial tool, well-quasi-orderings
appear frequently in varying fields of computer science,
ranging from graph theory to number
theory
\cite{higman1952ordering,kruskal1972theory,kvrivz1990well,Daligault2010}.
Well-quasi-orderings have also been highly successful in
proving the termination of verification algorithms.
One critical application of well-quasi-orderings
is to the verification of infinite state transition systems,
via the study of so-called Well-Structured Transition Systems (WSTS)
\cite{abdulla_general_1996,abdulla_verifying_1998,WQOInComputerScience,finkel_well_structured_2001}.



\subparagraph*{Noetherian spaces.}
\AP
One major roadblock arises when using \kl{well-quasi-orders}:
the powerset of a well-quasi-order may fail to be
one itself \cite{rado_1954}.
This is particularly problematic in the study of
\kl{WSTS}, where the powerset construction
appears frequently
\cite{jancar_note_1999,segoufin_bottom_up_2017,abdulla_general_1996}.
To tackle this issue, one can justify that the quasi-orders
of interest are not pathological, and are actually better
quasi-orders \citep*{pouzet1972bel, milner1985basic}.
Another approach is offered
by the topological notion of \emph{Noetherian space}, which as pointed
out by \citeauthor{goubault2007noetherian},
can act as a suitable generalisation of well-quasi-orderings
that is preserved under the powerset
construction~\cite{goubault2007noetherian}.

The topological analogues to WSTS enjoy similar decidability
properties, and there even exists
an analogue to Karp and Miller's forward analysis
for Petri nets~\cite{goubault2012verif}. Moreover, their
topological nature allows to verify systems
beyond the reach of quasi-orderings,
such as lossy concurrent polynomial programs \cite{goubault2012verif}.
This is possible because
the polynomials are handled via results from algebraic geometry,
through the notion of the \intro{Zariski topology}
over $\mathbb{C}^n$ \cite[Exercise 9.7.53]{goubault2013non}.

One drawback of the topological approach is that
many topologies correspond to a single quasi-ordering.
Hence, when the problem is better described via an ordering,
one has to choose a specific topology,
and there usually does not exist a finest one that is Noetherian.

\subparagraph*{Inductively defined datatypes.}
\AP
As for \kl{well-quasi-orders}, \kl{Noetherian spaces} are stable
under finite products
and finite sums \cite{lectureNotesWqo,goubault2013non}.
While this can be enough
to describe the set of configurations of a Petri net using $\Nat^k$,
it does not allow to talk about more complex data structures such as
channels, lists, or trees.

In the realm of well-quasi-orderings,
the specific cases of finite words and finite trees
are handled respectively via Higman's Lemma~\cite{higman1952ordering}
and Kruskal's Tree Theorem~\cite{kruskal1972theory}.
Let us recall that a word $u$ \intro[word embedding]{embeds} into a word $w$
(written $u \intro*\HigLeq v$) whenever
whenever there exists a strictly increasing map
$h \colon |w| \to |w'|$ such that
$w_i \leq w_{h(i)}$ for $1 \leq i \leq |w|$.
Similarly, a tree $t$ \intro[tree embedding]{embeds} into a tree $t'$
(written $t \intro*\KruLeq t'$)
whenever there exists a map
from nodes of $t$ to nodes of $t'$ respecting the least common
ancestor relation,
and increasing the colours of the nodes.
Proofs that finite words and finite trees
preserve \kl{well-quasi-orderings}
typically rely on a so-called
\intro{minimal bad sequence argument} due
to \citet*{nash_williams_wqo_1965}.
However, the argument is quite subtle, and
needs to be handled with care \cite{gallier_ann_1997,Singh2013SimplifiedPO}.
In addition, the 
argument is not compositional
and has to be slightly modified
whenever a new inductive construction is desired
\cite[e.g.,][]{dershowitz_gap_2003,Daligault2010}.

This picture has been
adapted to the topological setting
by proposing analogues of the \kl{word embedding} and \kl{tree embedding},
together with a proof that they preserve Noetherian spaces
\cite[Section 9.7]{goubault2013non}.
However, both the definitions and the proofs
have an increased complexity, as they rely on an adapted
``topological minimal bad sequence argument'' that appears
to be even more subtle.

\medskip
One could expect the situation to be more regular.
In an ML-like language, one can define words over an alphabet of type \verb|'a|
via a type declaration of the form
\texttt{'a word = Nil | Cons of 'a * 'a word},
and trees over an alphabet of type
\texttt{'a} via \texttt{'a tree = Node of 'a * ('a tree list)}.
In a more set-theoretical mindset,
one would write \verb|Nil| as the singleton set
$\intro*\singleton \defined \set{ \star }$,
$A + B$ the disjoint union of $A$ and $B$,
and $A \times B$ their product.
An inductive type would then be defined via a
least fixed point operator:
$\intro*\lfp{X}{F(X)}$.
In this language,
$\words{\Sigma} \equiv \lfp{X}{\singleton + \Sigma \times X}$,
and $\trees{\Sigma} \equiv \lfp{X}{\Sigma \times \words{X}}$.
In the case of well-quasi-orderings,
two generic fixed point constructions have already been
proposed~\cite{HASEGAWA2002113,freund2020kruskal}.
In these frameworks, the constructor $F$ in $\lfp{X}{F(X)}$
has to be a ``well-behaved'' functor of quasi-orders
in order for $\lfp{X}{F(X)}$ to be a \kl{well-quasi-order}.
Both proposals, while relying on different categorical
notions, successfully recover Higman's word embedding
and Kruskal's tree embedding via the least fixed point definitions
of words and trees. As a side effect, they reinforce
the idea that the two quasi-orders are somehow canonical.

In the case of Noetherian spaces, no equivalent framework exists
to build inductive datatypes,
and the notions of ``well-behaved'' constructors from~\cite{HASEGAWA2002113,freund2020kruskal} rule out
the use of
important Noetherian spaces, as they require that an element
$a \in F(X)$ has been built using \emph{finitely many}
elements of $X$: while this is the case for finite words and finite trees,
it does not hold for instance for the arbitrary powerset.
Moreover, there have been recent advances in placing Noetherian
topologies over spaces that are not straightforwardly obtained
through ``well-behaved''
definitions, such as infinite words~\cite{goubaultlarrecq2021infinitary},
or even ordinal length words~\cite{goubaultlarrecq2022infinitary}.

\subsection{Contributions of this paper}
\AP
In this paper, we propose a least fixed point theorem
for Noetherian topologies.
The main contribution of this paper is to
build topologies defined inductively over a set $X$.
This is done in a way that greatly differs from
the categorical frameworks proposed in the
study of \kl{well-quasi-orders}~\cite{HASEGAWA2002113,freund2020kruskal}, as
the construction of the space is entirely
\emph{decoupled} from the construction of the topology.
In particular, the set $X$ itself need not be inductively
defined.

In this setting, we consider a fixed set $X$
and a map $R$ from topologies $\TP$
over $X$ to topologies $R(\TP)$ over $X$.
Because the set of topologies over $X$ is a complete lattice,
it suffices to ask for $R$ to be monotone to guarantee
that it has a least fixed point, that we write
$\lfp{\TP}{R(\TP)}$.
In general, this least fixed point will not be Noetherian,
but we show that a simple sufficient condition on $F$ guarantees
that it is.
This main theorem (\cref{cor:gen:noethpres}),
encapsulates all the complexity of the
topological adaptations of the minimal bad
sequences arguments \citep[Section 9.7]{goubault2013non}, and we believe
that it has its own interest.

The necessity to separate the construction of the set
of points
from the construction of the topology might be perceived as
a weakness of the theory, when it is in fact a strength of our approach.
We illustrate this by
giving a shorter proof that the words of ordinal length
are Noetherian \cite{goubaultlarrecq2022infinitary},
without providing an inductive definition of the space.
As an illustration of the versatility of our framework,
we introduce a reasonable topology
over ordinal branching trees (with finite depth), and prove
that it is Noetherian using the same technique.

In the specific cases where the space of interest
can be obtained as a least fixed point of a ``well-behaved''
functor, we show how
\cref{cor:gen:noethpres} can be used to
generalise the categorical framework
of \citet*{HASEGAWA2002113} to a topological setting.


\subparagraph*{Outline.}
In \cref{sec:primer} we recall some of the main
results in the theory of \kl{Noetherian spaces}.
In \cref{sec:generic} we prove our main
result (\cref{cor:gen:noethpres}). In \cref{sec:transfinite}
we explore how this result covers existing topological results
in the literature, and provide a new non-trivial Noetherian space
(\cref{def:ft:alpha-trees}).
In \cref{sec:inductive}, we leverage our main result
to devise a Noetherian topology over inductively defined datatypes
(\cref{thm:induc:divisibility}),
and prove that this generalises the work of \citeauthor*{HASEGAWA2002113}
over well-quasi-orders (\cref{thm:induc:coincidence-alexandroff}).

\section{A Quick Primer on Noetherian Topologies}
\label{sec:primer}

A \intro*\kl{topological space}
is a pair $(\SP, \TP)$ where $\TP \subseteq \P(X)$,
$\TP$ is stable under finite intersections, and $\TP$
is stable under arbitrary unions.
Before formally introducing the topological counterpart to \kl{well-quasi-orderings},
let us provide a small dictionary from topology to orders.
Given a quasi-ordered set $(\QO, \leq)$,
a set $U$ is \intro{upwards-closed} whenever
$x \in U$ and $x \leq y$ implies $y \in U$ for every
$x,y \in X$.

\begin{definition}
	Let $(\QO, \leq)$ be a quasi-order.
	The \intro{Alexandroff topology}
	$\intro*\Alex{\leq}$
	over $\QO$ is the collection of \kl{upwards-closed} subsets of $\QO$.
\end{definition}

\begin{definition}
	Let $(\SP, \TP)$ be a topological space.
	The \intro{specialisation preorder}
	$\intro*\spec{\TP}$ is defined via
	$x \spec{\TP} y$ whenever for every open $U \in \TP$,
	if $x \in U$ then $y \in U$.
\end{definition}

It is an easy check that
the \kl{specialisation pre-order} of the \kl{Alexandroff topology}
of a quasi-order $\leq$ is the quasi-order itself.
This allows to build intuition by getting back and forth
between topologies and quasi-orders.
Several topologies can share the same \kl{specialisation pre-order} $\leq$,
among those, the \kl{Alexandroff} topology is the finest.

\AP We can now build the topological analogue to \kl{wqos}
through the notion of compactness: a subset $K$ of
$(\SP, \TP)$ is defined as \intro{compact}
whenever
from every family $(U_i)_{i \in I}$ of open
sets such that $K \subseteq \bigcup_{i \in I} U_i$, one can extract
a finite
subset $J \subseteq_f I$ such that $K \subseteq \bigcup_{i \in J} U_i$.
A quasi-order $(\QO,\leq)$ is \kl{wqo} if and only if
every subset $K$ of $\QO$ is compact.
Generalising this property to arbitrary topological spaces $(\SP,\TP)$,
a topological space $(\SP,\TP)$ is said to be a
\intro{Noetherian space} whenever
every subset of $\SP$ is compact.

\begin{remark}
	A space $(X, \tau)$ is \kl{Noetherian} if and only
	if for every increasing sequence of open subsets $(U_i)_{i \in \mathbb{N}}$,
	there exists $j \in \mathbb{N}$ such that $\bigcup_{i \in \mathbb{N}} U_i
		= \bigcup_{i \leq j} U_i$.
\end{remark}

Following the ideas of \kl{wqos}, an algebra of \kl{Noetherian spaces}
has been developed and is described in
\cref{fig:intro:noeth-algebra}
\cite[see][]{goubault2007noetherian,goubault2013non,goubaultlarrecq2022infinitary}.

\begin{table}[t]
	\centering
	\iflncs
		\caption{An \intro(noeth){\emph{algebra}} of \kl{Noetherian spaces}.}
	\fi
	\begin{tabular}{lcl}
		\toprule
		\textbf{Constructor} & \textbf{Syntax}      & \textbf{Topology}                 \\
		\midrule
		Well-quasi-orders    & $\QO$                & \kl{Alexandroff topology}
		\\
		Complex vectors      & $\mathbb{C}^k$       & \kl{Zariski topology}
		\\
		\addlinespace
		Disjoint sum         & $\SP_1 + \SP_2$      & \kl{co-product topology}
		\\
		Product              & $\SP_1 \times \SP_2$ & \kl{product topology}
		\\
		\addlinespace
		Finite words         & $\words{\SP}$        & \kl{subword topology}
		\\
		Finite trees         & $\trees{\SP}$        & \kl{tree topology}
		\\
		Finite multisets     & $\fmSets{\SP}$       & \kl{multiset topology}
		\\
		\addlinespace
		Transfinite words    & $\SP^{<\alpha}$      & \kl{transfinite subword topology}
		\\
		Powerset             & $\pSets{\SP}$        & \kl{Lower-Vietoris}
		\\
		\bottomrule
	\end{tabular}
	\iflipics
		\caption{An \intro(noeth){\emph{algebra}} of \kl{Noetherian spaces}.}
	\fi
	\label{fig:intro:noeth-algebra}
	\label{claim:join-prod-noeth}
\end{table}

\section{Refinements of Noetherian topologies}
\label{sec:generic}
\AP
Let us fix a set $X$
equipped the \intro{trivial topology}
$\intro*\TrivTop \defined \{ \emptyset, X \}$.
This space is \kl{Noetherian} because there are finitely many open
sets.
The approach taken in this paper
is to iteratively refine this topology
while keeping it Noetherian, and ultimately prove
that the \emph{limit} of this construction remains Noetherian.

\begin{definition}
	A \intro{refinement function} over a set $X$
	is a function
	$\intro*\ARef$ mapping topologies over $X$
	to topologies over $X$. Moreover, we assume that
	$\ARef(\TP)$ is \kl{Noetherian} whenever $\TP$ is,
	and that $\ARef(\TP) \subseteq \ARef(\TP')$ when
	$\TP \subseteq \TP'$.
\end{definition}

The collection of topologies over a set $X$
is itself a set, and forms a complete lattice
for inclusion. Thanks to Tarski's fixed point theorem,
every \kl{refinement function} $\ARef$ has a least fixed point,
which can be obtained
by transfinitely iterating $\ARef$ from the \kl{trivial topology}.
Let us write $\lfp{\TP}{\ARef(\TP)}$ for the least fixed point of
$\ARef$.

\AP Given a quasi-order $(\QO, \leq)$ and a set
$E \subseteq \QO$, let us define the \intro{upwards-closure} of $E$,
written $\upset{\leq}{E}$, as the set of elements that are
greater or equal than some element of  $E$ in $\QO$.

\begin{example}[Natural Numbers]
	\label{ex:gen:natural-numbers}
	Over $X \defined \Nat$, one can define
	$\mathsf{Div}(\TP)$ as the topology
	generated by the sets
	$\upset{\leq}{(U + 1)}$ for $U \in \TP$.
	Then $\mathsf{Div}(\TrivTop) = \set{\emptyset,
			\upset{\leq}{1}, \Nat}$,
	$\mathsf{Div}^2(\TrivTop) = \set{\emptyset,
			\upset{\leq}{1}, \upset{\leq}{2}, \Nat}$. More generally,
	for every $k \geq 0$,
	$\mathsf{Div}^k(\TrivTop) = \set{
			\emptyset, \upset{\leq}{1}, \dots, \upset{\leq}{k},
			\Nat
		}$.
	It is an easy check that
	$\lfp{\TP}{\mathsf{Div}(\TP)}$ is precisely $\Alex{\leq}$,
	which is \kl{Noetherian} because $(\Nat, \leq)$
	is a \kl{well-quasi-ordering}.
\end{example}

Not all \kl{refinement functions} behave as nicely
as in \cref{ex:gen:natural-numbers}, and one can obtain
non-\kl{Noetherian} topologies via their least fixed points.

\subsection{An ill-behaved refinement function}
\label{sec:sub:ill-behaved}

Let us consider for this section $\Alpha \defined \{a,b\}$ with the
\kl{discrete topology}, i.e., $\set{\emptyset, \{a\}, \{b\}, \Alpha}$.
Let us now build the set $\words{\Alpha}$ of finite words
over $\Sigma$.
Whenever $U$ and $V$ are subsets of $\words{\Alpha}$,
let us write $UV$ for their concatenation, defined as $\setof{uv}{u \in U, v \in V}$.
\begin{definition}
	Let $\intro*\BadIterator$, be the function
	mapping a topology $\TP$ over $\words{\Alpha}$
	to the topology generated by the sets
	$UV$ where $U \subseteq \Sigma$ and $V \in \TP$,
\end{definition}

We refer to \cref{fig:gen:prefixtopo}
for a graphical presentation of the first two iterations
of the refinement function $\BadIterator$.
For the sake of completeness, let us compute $\lfp{\TP}{\BadIterator(\TP)}$,
which is the \kl{Alexandroff topology} of the prefix ordering on words.
Beware that this is not the usual notion of ``prefix topology''
in the literature
\cite[see][resp. Section 8 and Exercice
	9.7.36]{finkel_forward_2020,goubault2013non}.

\begin{definition}
	The \intro{prefix topology} $\intro*\PrefTopo$, over $\words{\Sigma}$ is
	generated by the following open sets:
	$U_1 \dots U_n \words{\Sigma}$, where $n \geq 0$ and $U_i \subseteq \Sigma$.
\end{definition}
\vspace{-2ex}
\begin{restatable}{lemma}{lemgenpreflfp}\label{lem:gen:pref-lfp}
	The \kl{prefix topology} over $\words{\Sigma}$ is the least fixed point
	of $\BadIterator$.
\end{restatable}
\begin{sendappendix}
	\ifsubmission \lemgenpreflfp* \fi
	\begin{proof}\label{proof:gen:pref-lfp}
		Consider a subbasic open set $W \in \BadIterator(\PrefTopo)$.
		It is of the form $U V$ with $U \subseteq
			\Sigma$ and $V \in \PrefTopo$. Hence,
		$UV \in \PrefTopo$. We have proven that,
		$\BadIterator(\PrefTopo) \subseteq \PrefTopo$.

		Conversely, consider a subbasic open set $W \in \PrefTopo$. Either it is
		$\emptyset$, or $\words{\Sigma}$, in which case
		it trivially belongs to $\lfp{\TP}{\BadIterator(\TP)}$. Or it is
		of the form $U_1 \dots U_n \words{\Sigma}$, with $U_i \subseteq \Sigma$ for
		$1 \leq i \leq n$, in which case one proves by induction over
		$n$ that it belongs to $\BadIterator^n(\TrivTop)$. \qedhere
	\end{proof}
\end{sendappendix}

\begin{lemma}
	The function $\BadIterator$ is a \kl{refinement function}.
\end{lemma}
\begin{proof}
	It is an easy check that whenever $\TP \subseteq \TP'$,
	$\BadIterator(\TP) \subseteq \BadIterator(\TP')$.
	Now, assume that $\TP$ is \kl{Noetherian},
	it remains to prove that $\BadIterator(\TP)$ remains Noetherian.
	Consider a subset $E \subseteq \words{\Sigma}$ and let us prove that $E$
	is compact in $\BadIterator(\TP)$.

	For that, we consider
	an open cover $E \subseteq \bigcup_{i \in I} W_i$, where
	$W_i \in \BadIterator(\TP)$.
	Thanks to Alexander's subbase lemma, we can assume without loss
	of generality that $W_i$ is a subbasic open set of $\BadIterator(\TP)$, that is,
	$W_i = U_i V_i$ with $U_i \subseteq \Sigma$ and $V_i \in \TP$.

	Since $(\words{\Sigma}, \TP) \times (\words{\Sigma}, \TP)$
	is Noetherian (see \cref{claim:join-prod-noeth}),
	there exists a finite set $J \subseteq I$ such that
	$\bigcup_{i \in J} U_i \times V_i = \bigcup_{i \in I} U_i \times V_i$.
	This implies that
	$E \subseteq \bigcup_{i \in J} U_i V_i$,
	and provides
	a finite subcover of $E$.
	\qedhere
\end{proof}

\begin{figure}[t]
	\centering
	\resizebox{0.7\columnwidth}{!}{
		\begin{tikzpicture}
			\begin{scope}
				\node (top) at (0,4) {$\Sigma^*$};
				\node (bot) at (0,0) {$\emptyset$};
				\draw[->] (bot) -- (top);
			\end{scope}
			\begin{scope}[xshift=3cm]
				\node (top) at (0,4) {$\Sigma^*$};
				\node (bot) at (0,0) {$\emptyset$};
				\node (aS) at (-1,2) {$a \Sigma^*$};
				\node (bS) at (1,2) {$b \Sigma^*$};
				\draw[->] (bot) -- (aS) -- (top);
				\draw[->] (bot) -- (bS) -- (top);
			\end{scope}
			\begin{scope}[xshift=9cm]
				\node (top) at (0,4) {$\Sigma^*$};
				\node (bot) at (0,0) {$\emptyset$};
				\node (aS) at (-1,2.5) {$a \Sigma^*$};
				\node (bS) at (1,2.5) {$b \Sigma^*$};
				\node (aaS) at (-3,1.5) {$aa\Sigma^*$};
				\node (abS) at (-1,1.5) {$ab\Sigma^*$};
				\node (baS) at (1,1.5) {$ba\Sigma^*$};
				\node (bbS) at (3,1.5) {$bb\Sigma^*$};
				\draw[->] (bot) -- (aaS) -- (aS) -- (top);
				\draw[->] (bot) -- (bbS) -- (bS) -- (top);
				\draw[->] (bot) -- (baS) -- (bS);
				\draw[->] (bot) -- (abS) -- (aS);
			\end{scope}
		\end{tikzpicture}
	}
	\caption{Iterating $\BadIterator$ over $\Sigma^*$. On the left the
		\kl{trivial topology} $\TrivTop$, followed by $\BadIterator$, and on the right
		$\BadIterator^2$.}
	\label{fig:gen:prefixtopo}
	\vspace{-2ex}
\end{figure}

The sequence
$\bigcup_{0 \leq i \leq k} a^ib\Sigma^*$, for $k \in \Nat$,
is a strictly increasing sequence of opens. Therefore,
the \kl{prefix topology} is not \kl{Noetherian}.
The terms $a^ib\Sigma^*$
can be observed in \cref{fig:gen:prefixtopo}
as a diagonal of incomparable open sets.

\begin{corollary}
	The topology $\lfp{\TP}{\BadIterator(\TP)}$ is not \kl{Noetherian}.
\end{corollary}

\AP
The \kl{prefix topology} is not \kl{Noetherian}, even when starting
from a finite alphabet. However, we claimed
in \cref{sec:intro} that there is a natural
generalisation of the \kl{subword embedding} to
topological spaces which is Noetherian.
Before introducing this topology, let us write
$\intro*\openword{U_1, \dots, U_n}$ as a shorthand
notation for the set
$\Sigma^* U_1 \Sigma^* \dots \Sigma^* U_n \Sigma^*$.

\begin{definition}[{Subword topology \cite[Definition 9.7.26]{goubault2013non}}]
	Given a topological space $(\Sigma,\tau)$, the space $\Sigma^*$ of finite
	words over $\Sigma$ can be endowed with the \intro{subword
		topology}, generated by the open sets
	$\openword{U_1, \dots, U_n}$ when $U_i \in \tau$.
\end{definition}

The \emph{topological Higman lemma}~\cite[Theorem 9.7.33]{goubault2013non}
states that the \kl{subword topology} over $\words{\Sigma}$ is \kl{Noetherian}
if and only if $\Sigma$ is Noetherian.
Let us now reverse engineer a \kl{refinement function}
whose least fixed point is the subword topology.

\begin{definition}
	Let $(\Sigma, \theta)$ be a topological space.
	Let $\intro*\RegSubExp$ be defined as
	mapping a topology $\TP$ over $\words{\Sigma}$ to the topology
	generated by the following sets:
	\begin{inparaitem}[]
		\item $\upset{\HigLeq}{UV}$ for $U,V \in \tau$;
		\item and $\upset{\HigLeq}{W}$, for $W \in \theta$.
	\end{inparaitem}
\end{definition}
\vspace{-2ex}
\begin{restatable}{lemma}{regsubexplfp}
	\label{lem:gen:regsubexp-lfp}
	Let $(\Sigma, \theta)$ be a topological space.
	The \kl{subword topology} over $\words{\Sigma}$
	is the least fixed point of $\RegSubExp$.
\end{restatable}
\begin{sendappendix}
	\ifsubmission \regsubexplfp* \fi
	\begin{proof}
		First, we notice that the  subword topology
		is stable under $\RegSubExp$.
		Then, we prove
		by induction on $n$ shows that $\openword{U_1, \dots, U_n}$
		is open in the least fixed point of $\RegSubExp$.
		\qedhere
	\end{proof}
\end{sendappendix}

It is an easy check that the \kl{subword topology}
over $\words{\Sigma}$ is the least fixed point of $\RegSubExp$.
In order to show that $\RegSubExp$ is a \kl{refinement function},
we first claim that the two parts of the topology can be dealt with
separately.
\begin{restatable}[{\cite[Proposition 9.7.18]{goubault2013non}}]{lemma}{lemjoinnoeth}
	\label{lem:gen:join-noeth}
	If $(\SP, \TP)$ and $(\SP, \TP')$ are \kl{Noetherian},
	then $\SP$ endowed the topology generated by $\TP \cup \TP'$ is Noetherian.
\end{restatable}
\begin{sendappendix}
	\ifsubmission \lemjoinnoeth* \fi
	\begin{proof}
		The space $\SP$ endowed with the topology generated by
		$\TP \cup \TP'$ is a quotient of $(\SP, \TP) + (\SP, \TP')$.
		Therefore, it is Noetherian \cite[Proposition 9.7.18]{goubault2013non}.
	\end{proof}
\end{sendappendix}

\begin{lemma}
	\label{lem:gen:reg-exp-refinement}
	Let $(\Sigma, \theta)$ be a \kl{Noetherian} topological space.
	The map $\RegSubExp$ is a \kl{refinement function} over $\Sigma$.
\end{lemma}
\begin{proof}
	We leave the monotonicity of $\RegSubExp$ as an exercice
	and focus on the proof that $\RegSubExp(\TP)$ is \kl{Noetherian},
	whenever $\TP$ is.
	Thanks to \cref{lem:gen:join-noeth}, it suffices to prove that the topology
	generated by the sets $\upset{\HigLeq}{UV}$ ($U,V$ open in $\tau$),
	and the topology generated by the sets $\upset{\HigLeq}{W}$ ($W$ open in
	$\theta$) are Noetherian.

	Let $(\upset{\HigLeq}{U_i V_i})_{i \in \Nat}$ be a sequence of
	open sets. Because Noetherian topologies are closed under
	products (\cref{claim:join-prod-noeth}),
	the sequence $(\bigcup_{i \leq k} U_i \times V_i)_{k \in \Nat}$
	is asymptotically constant. Hence,
	the sequence $\bigcup_{i \leq k} \upset{\HigLeq}{U_i V_i}$
	also is.

	Let $\upset{\HigLeq}{W_i}$ be a sequence of open sets.
	Because $\theta$ is Noetherian, the sequence $\bigcup_{i \leq k} W_i$
	is asymptotically constant,
	hence so is the sequence $\bigcup_{i \leq k} \upset{\HigLeq}{W_i}$.
	\qedhere
\end{proof}

We have designed two \kl{refinement functions}
$\BadIterator$ and $\RegSubExp$ over $\words{\Sigma}$. The least
fixed point of the former is not \kl{Noetherian}, as opposed to the
least fixed point of the latter.
We have depicted the result of iterating $\RegSubExp$ twice
over the \kl{trivial topology} in
\cref{fig:induct:correctingprefix}. As opposed to
$\BadIterator$, the ``diagonal'' elements are comparable for inclusion.
To further elaborate the difference between $\BadIterator$ and
$\RegSubExp$, let us compare their behaviour with respect to
subsets of $\words{\Sigma}$.

Let $V \defined a \words{\Sigma}$,
which is a closed subset of $(\words{\Sigma}, \BadIterator(\TrivTop))$.
When endowing $V$ with the topology
induced by $\BadIterator(\TrivTop)$,
we obtain the space $(V, \TrivTop)$.
When endowing $V$ with the topology induced by
$\BadIterator^2(\TrivTop)$,
we obtain a space
$(V, \set{ \emptyset, aa\words{\Sigma}, ab\words{\Sigma}, V })$.
However, if one considers $V$ as a topological space
itself, then applying $\BadIterator$ over $(V, \TrivTop)$
leads to the open sets $\set{\emptyset, aa \words{\Sigma}, V}$,
which is a different topology.

Let $W \defined \words{\Sigma} a \words{\Sigma}$,
which is a closed in $\RegSubExp(\TrivTop)$.
As for $V$, the topology induced on $W$
by $\RegSubExp(\TrivTop)$ is the trivial topology.
However, when considering $(W, \TrivTop)$
as a topological space, we obtain
the same topology over $W$ whether we
build the topology induced by $\RegSubExp^2(\TrivTop)$,
or apply $\RegSubExp$ to $W$ itself.

\begin{figure}[t]
	\centering
	\begin{tikzpicture}
		\node (top) at (0,4) {$\Sigma^*$};
		\node (bot) at (0,0) {$\emptyset$};
		\node (aS) at (-1,2.5) {$\Sigma^* a \Sigma^*$};
		\node (bS) at (1,2.5) {$\Sigma^* b \Sigma^*$};
		\node (aaS) at (-3,1.5) {$\Sigma^* a \Sigma^* a\Sigma^*$};
		\node (abS) at (-1,1.5) {$\Sigma^* a\Sigma^* b\Sigma^*$};
		\node (baS) at (1,1.5) {$\Sigma^* b\Sigma^* a\Sigma^*$};
		\node (bbS) at (3,1.5) {$\Sigma^* b\Sigma^* b\Sigma^*$};
		\draw[->] (bot) -- (aaS) -- (aS) -- (top);
		\draw[->] (bot) -- (bbS) -- (bS) -- (top);
		\draw[->] (bot) -- (baS) -- (bS);
		\draw[->] (bot) -- (abS) -- (aS);
		\draw[->, thick, red] (baS) -- (aS);
		\draw[->, thick, red] (abS) -- (bS);
	\end{tikzpicture}
	\caption{The topology $\RegSubExp^2(\TrivTop)$, with bold red arrows
		for the inclusions that were not present between the
		``analogous sets'' in $\BadIterator^2(\TrivTop)$.}
	\label{fig:induct:correctingprefix}
\end{figure}

\subsection{Well-behaved refinement functions}
\label{sec:sub:well-behaved}
\AP

As hinted in the previous section,
the behaviour of the \kl{refinement function} with respect to subsets
will act as a sufficient condition to separate the well-behaved
ones from the others. In order to make the idea of computing
the refinement function directly over a subset precise,
we will replace a subset with the induced topology by
a ``restricted'' topology over the whole space.

\begin{definition}
	Let $(\SP,\TP)$ be a topological space
	and $H$ be a closed subset of $\SP$.
	Define the \intro{subset restriction}
	$\intro*\restr{\TP}{H}$
	to be the topology generated
	by the opens $U \cap H$ where $U$ ranges over
	$\TP$.
\end{definition}

Let $\SP$ be a topological space, and $H$ be a proper closed subset of $\SP$.
The space $\SP$ endowed with $\restr{\TP}{H}$
has a lattice of open sets that is isomorphic to
the one of the space $H$ endowed with the topology induced by $\TP$,
except for the entire space $\SP$ itself.
Beware that, the two spaces are in general not homeomorphic.

\begin{example}
	Let $\mathbb{R}$ be endowed with the usual metric topology.
	The set $\{ a \}$ is a closed set.
	The induced topology over $\{a\}$ is $\set{\emptyset, \{a\}}$.
	The \kl{subset restriction} of the topology to $\{a\}$
	is $\TP_a \defined \set{\emptyset, \{ a \}, \mathbb{R}}$.
	Clearly, $(\mathbb{R}, \TP_a)$ and $(\{a\}, \TrivTop)$
	are not homeomorphic.
\end{example}

In order to build intuition, let us
consider the special case of an \kl{Alexandroff topology}
over $X$ and compute the \kl{specialisation preorder}
of $\restr{\tau}{H}$, where $H$ is a downwards closed set.
\ifsubmission The proof can be found in \cref{proof:gen:specpreorder}
	page \pageref{proof:gen:specpreorder}.\fi

\begin{restatable}{lemma}{specPreorderRestrictions}
	\label{lem:gen:specpreorder}
	Let $\tau = \Alex{\leq}$ over a set $X$,
	and $x,y \in X$.
	Then,
	$x \spec{\restr{\tau}{H}} y$ if and only
	if $x \spec{\tau} y \in H$ or $x \not \in H$.
\end{restatable}
\begin{sendappendix}
	\ifsubmission
		\specPreorderRestrictions*
	\fi
	\begin{proof}
		\label{proof:gen:specpreorder}
		Let us write $\uparrow F$ for the set of points that are above $F$
		for $\leq$, and $\uparrow x$ as a shorthand notation for $\uparrow \{ x \}$,
		the set of points above $x$. Let us now unpack
		the definition of $x \spec{\restr{\tau}{H}} y$.
		\begin{align*}
			x \spec{\restr{\tau}{H}} y
			 & \iff \forall U \in \restr{\tau}{H}, x
			\in U \Rightarrow y \in U                  \\
			 & \iff \forall U \in \tau, x \in U \cap H
			\Rightarrow y \in U \cap H                 \\
		\end{align*}
		Let $x \spec{\restr{\tau}{H}} y$. If $x \in H$, then
		for every open set $U \in \tau$, $x \in U \cap H$, hence
		$y \in U \cap H$. As a consequence,
		$x \spec{\tau} y$ and both belong to $H$.

		Conversely, assume that $x \not \in H$, then
		$x \in U \cap H \implies y \in U \cap H$ vacuously for every
		$U \in \tau$, hence $x \spec{\restr{\tau}{H}} y$.
		Whenever, $x \spec{\tau} y \in H$,
		then $x \not \in H$ implies $y \not \in H$, which is absurd.
		Therefore, $x \spec{\restr{\tau}{H}} y$.
		\qedhere
	\end{proof}
\end{sendappendix}


\begin{definition}
	\label{def:topo-expander}
	A \intro{topology expander} is a
	\kl{refinement function} $\intro*\AExp$ that satisfies the following extra
	property:
	for every \kl{Noetherian topology} $\tau$
	satisfying $\tau \subseteq \AExp(\tau)$,
	for all closed set $H$ in $\tau$,
	$\restr{\AExp(\tau)}{H} =
		\restr{\AExp(\restr{\tau}{H})}{H}$.
	We say that $\AExp$ respects subsets.
\end{definition}

\begin{note}
	In
	\cref{def:topo-expander}, the equality can be
	replaced by: $H$ is closed in $\AExp(\restr{\tau}{H})$ and
	$\restr{\AExp(\tau)}{H} \subseteq
		\restr{\AExp(\restr{\tau}{H})}{H}$.
\end{note}

As proven at the end of \cref{sec:sub:ill-behaved}, $\BadIterator$
fails to be a \kl{topology expander}. Let us quickly
prove that $\RegSubExp$ is a topology expander.

\begin{lemma}
	\label{lem:gen:regsubexp-topo-expander}
	Let $(\Sigma, \theta)$ be a \kl{Noetherian space}.
	Then $\RegSubExp$ is a \kl{topology expander}.
\end{lemma}
\begin{proof}
	We have proven in
	\cref{lem:gen:reg-exp-refinement} that $\RegSubExp$
	is a \kl{refinement function}. Let us now prove that
	it respects subsets.

	Let $\tau$ be a \kl{Noetherian} topology
	over $\words{\Sigma}$, such that
	$\tau \subseteq \RegSubExp(\tau)$.
	Let $H$ be a closed subset of $(\words{\Sigma}, \tau)$.
	Notice that as $H$ is closed in $\tau$, and since $\tau \subseteq
		\RegSubExp(\tau)$, $H$ is \kl{downwards closed} for $\HigLeq$.
	As a consequence,
	$(\upset{\HigLeq}{UV}) \cap H
		= (\upset{\HigLeq}{(U \cap H) (V\cap H)) \cap H}$.
	Similarly,
	$(\upset{\HigLeq}{W}) \cap H
		= (\upset{\HigLeq}{(W \cap H)}) \cap H$.
	Hence, $\restr{\RegSubExp(\tau)}{H} \subseteq
		\restr{\RegSubExp(\restr{\tau}{H})}{H}$.
	\qedhere
\end{proof}

\subsection{Iterating Expanders}

Our goal is now to prove that
\kl{topology expanders} are
\kl{refinement functions} that can be safely iterated.
For that, let us first define precisely what ``iterating
transfinitely'' a refinement function means.

\begin{definition}
	Let $(\SP,\TP)$ be a topological space,
	and $\AExp$
	be a \kl{topology expander}.
	The \intro*\kl{limit topology} $\AExp^\alpha(\tau)$
	is defined as: $\tau$ when $\alpha = 0$,
	$\AExp(\AExp^\beta(\tau))$ when $\alpha = \beta + 1$,
	and as the join of the topologies
	$\AExp^\beta(\tau)$ for all $\beta < \alpha$, when
	$\alpha$ is a limit ordinal.
\end{definition}

We devote the rest of this section
to proving our main theorem,
which immediately implies that
least fixed points of \kl{topology expanders}
are \kl{Noetherian}. Notice that the theorem
is trivial whenever $\alpha$ is a successor ordinal.

\begin{restatable}{proposition}{limittopologies}
	\label{thm:gen:limittoponoeth}
	Let $\alpha$ be an ordinal, $\TP$ be a topology, and
	$\AExp$ be a \kl{topology expander}.
	If $\AExp^\beta(\tau)$ is \kl{Noetherian} for all $\beta < \alpha$,
	and $\tau \subseteq \AExp(\tau)$,
	then $\AExp^\alpha(\tau)$ is \kl{Noetherian}.
\end{restatable}
\begin{theorem}[Main Result]
	\label{cor:gen:noethpres}
	Let $X$ be a set
	and $\AExp$ be a \kl{topology expander}.
	The least fixed point of $\AExp$
	is a \kl{Noetherian topology} over $X$.
\end{theorem}

\subsubsection{The topological minimal bad sequence argument.}
\label{sec:sub:iteration}
\AP
In order to define what a minimal bad sequence is,
we first introduce a well-founded partial ordering over
the elements of $\AExp^{\alpha}(\tau)$.
With an open set $U \in \AExp^{\alpha}(\tau)$, we associate
a depth $\intro*\depth(U)$,
defined as the smallest ordinal $\beta \leq \alpha$ such that
$U \in \AExp^{\beta}(\tau)$.
We then define $U \intro*\osubeq V$ to hold
whenever $\depth(U)\leq \depth(V)$,
and $U \intro*\osub V$ whenever $\depth(U) < \depth(V)$.
It is an easy check that this is a well-founded partial order
over $\AExp^{\alpha}(\tau)$.

As a first step towards proving that
$\AExp^{\alpha}(\tau)$ is \kl{Noetherian} for a limit ordinal $\alpha$,
we first reduce the problem to
opens of depth strictly less than $\alpha$ itself.

\begin{restatable}{lemma}{lemgendepthalpha}
	\label{lem:gen:depth-alpha}
	Let $\alpha$ be a limit ordinal, and $\AExp$ be a \kl{topology expander}.
	The topology $\AExp^\alpha(\tau)$
	has a subbasis of elements of depth strictly
	below $\alpha$.
\end{restatable}
\begin{sendappendix}
	\ifsubmission \lemgendepthalpha* \fi
	\begin{proof}
		By definition of the limit topology.
	\end{proof}
\end{sendappendix}

Let us recall the notion of topological bad sequence
designed by \citet*[Lemma 9.7.31]{goubault2013non}
in the proof of the Topological Kruskal Theorem,
adapted to our ordering of subbasic open sets.
This notion of bad sequence is tailored to mimic the notion
of \kl(wqo){good sequences} and \kl(wqo){bad sequences}
in \kl{well-quasi-orderings}.

\begin{definition}
	Let $(\SP,\tau)$ be a topological space.
	A sequence $\intro*\AUseq = {(U_i)}_{i \in \mathbb N}$
	of open sets is \intro(noeth){good}
	if there exists $i \in \mathbb{N}$
	such that $U_i \subseteq \bigcup_{j < i} U_j$.
	A sequence that is not \kl(noeth){good} is called \intro(noeth){bad}.
\end{definition}

\begin{restatable}{lemma}{lemgenminibad}
	\label{lem:gen:minbadseq}
	Let $\alpha$ be a limit ordinal, and $\AExp$
	be a \kl{topology expander}
	such that
	$\AExp^\alpha(\tau)$ is not
	\kl{Noetherian}. Then, there exists
	a \kl(noeth){bad sequence} $\AUseq$
	of opens in $\AExp^\alpha(\tau)$
	of \kl{depth} less than $\alpha$
	that is lexicographically minimal for $\osubeq$.
	Such a sequence is called \intro(noeth)[minimal bad sequence]{minimal bad}.
\end{restatable}
\begin{sendappendix}
	\ifsubmission \lemgenminibad* \fi
	\begin{proof}
		Assume that $\AExp^{\alpha}(\tau)$
		is not \kl{Noetherian}. There
		exists a sequence $(U_i)_{i \in \Nat}$
		of subbasic open sets that is \kl(noeth){bad}
		and lexicographically minimal with respect to $\osubeq$
		\cite[Lemma 9.7.31]{goubault2013non}.
	\end{proof}
\end{sendappendix}

We deduce that in a \kl{limit topology},
\kl(noeth){minimal bad sequences} are not allowed to
use opens of arbitrary \kl{depth}.

\begin{lemma}
	\label{lem:gen:depthminbad}
	Let $\alpha$ be a limit ordinal,
	$\tau$ be a topology and $\AExp$ be a \kl{topology expander}
	such that $\AExp^{\beta}(\tau)$ is \kl{Noetherian}
	for all $\beta < \alpha$.
	Assume that $\AUseq = (U_i)_{i \in \mathbb N}$
	is a \kl(noeth){minimal bad sequence}
	of $\AExp^\alpha(\tau)$. Then, for every $i \in \mathbb N$,
	$\depth(U_i)$ is either $0$ or a successor ordinal.
\end{lemma}
\begin{proof}
	Assume by contradiction that
	there exists $i \in \mathbb N$ such that
	$\depth(U_i)$ is a limit ordinal $d_i$.
	This proves that $U_i$ is obtained as a union of
	open sets in $\AExp^\beta(\tau)$ for $\beta < d_i$.
	Since $\AExp^{d_i}(\tau)$ is \kl{Noetherian},
	one can define $U_i$ as a finite union of
	open sets of depth less than $d_i$. As a consequence,
	$\depth(U_i) < d_i$, which is absurd.
	\qedhere
\end{proof}

\begin{definition}
	Let $\alpha$ be an ordinal,
	$\tau$ be a topology, $\AExp$ be a \kl{topology expander}
	such that $\tau \subseteq \AExp(\tau)$, and
	let $U \in \AExp^{\alpha}(\tau)$.
	The topology $\intro*\TDown(U)$ is
	generated by the open sets
	$V$ such that $V \subseteq U$,
	where $V$ ranges over $\AExp^{\alpha}(\tau)$.
\end{definition}
\vspace{-2ex}
\begin{restatable}{lemma}{lemgentdownu}
	\label{lem:gen:tdown-u}
	Let $\alpha$ be an ordinal,
	$\tau$ be a topology, $\AExp$ be a \kl{topology expander}
	such that $\tau \subseteq \AExp(\tau)$, and
	let $U \in \AExp^{\alpha}(\tau)$.
	If $\depth(U) = \gamma + 1$, and $\AExp^{\gamma}(\tau)$
	is Noetherian,
	then $U \in \AExp(\TDown(U))$.
\end{restatable}
\begin{sendappendix}
	\ifsubmission \lemgentdownu* \fi
	\begin{proof}
		Let $U \in \AExp^{\gamma + 1}(\tau) = \AExp(\AExp^{\gamma}(\tau))$.
		By definition, an open set $V \in \AExp^{\alpha}(\tau)$ satifies
		$\depth(V) < \depth(U)$ if and only if
		it belongs to $\AExp^{\gamma}(\tau)$.
		As a consequence, $\AExp^{\gamma}(\tau) = \TDown(U)$,
		and $U \in \AExp(\TDown(U))$.
	\end{proof}
\end{sendappendix}

If $\AUseq$ is a \kl(noeth){minimal bad sequence}
in $(X, \AExp^{\alpha}(\tau))$,
then $U_i \not \subseteq  \bigcup_{j < i} U_j \defined V_i $,
i.e.,
$U_i \cap V_i^c \neq \emptyset$.
We can now use our \kl{subset restriction}
operator to devise a topology
associated to this \kl(noeth){minimal bad sequence}.
Noticing that $H_i \defined V_i^c$ is a closed set
in $\AExp^{\alpha}(\tau)$, we
can build the subset restriction
$\restr{\TDown(U_i)}{H_i}$.

\begin{definition}
	Let $\alpha$ be an ordinal,
	$\tau$ be a topology, $\AExp$ be a \kl{topology expander}
	such that $\tau \subseteq \AExp(\tau)$, and
	let
	$\AUseq = (U_i)_{i \in \mathbb{N}}$
	be a \kl(noeth){minimal bad sequence}
	in $\AExp^{\alpha}(\tau)$.
	Then, the \intro(noeth){minimal topology}
	$\intro*\UTopo(\AExp^{\alpha}(\tau))$ is generated by
	$\bigcup_{i \in \mathbb{N}} \restr{\TDown(U_i)}{H_i}$,
	where
	$H_i \defined
		(\bigcup_{j < i} U_j)^c$.
\end{definition}
\vspace{-2ex}%
\begin{restatable}{lemma}{minimaltopology}
	\label{lem:gen:mintopo}
	Let $\alpha$ be an ordinal,
	$\tau$ be a topology, $\AExp$ be a \kl{topology expander}
	such that $\tau \subseteq \AExp(\tau)$, and
	let
	$\AUseq = (U_i)_{i \in \mathbb{N}}$
	be a \kl(noeth){minimal bad sequence}
	in $\AExp^{\alpha}(\tau)$.
	Then, the \kl(noeth){minimal topology}
	$\UTopo(\AExp^{\alpha}(\tau))$
	is \kl{Noetherian}.
\end{restatable}
\begin{proof}
	Assume by contradiction that
	$\UTopo(\AExp^\alpha(\tau))$ is not \kl{Noetherian}.
	Let us define $V_i$ as $\bigcup_{j < i} U_j$,
	and $H_i$ as $V_i^c$.

	Thanks to \cite[Lemma 9.7.15]{goubault2013non}
	there exists a \kl(noeth){bad sequence}
	$\mathcal{W} \defined (W_i)_{i \in \mathbb{N}}$
	of subbasic elements of $\UTopo(\AExp^\alpha(\tau))$.
	By definition,
	$W_i$ is in some $\restr{\TDown(U_j)}{H_j}$.
	Let us select a mapping $\rho \colon \mathbb N \to \mathbb N$,
	such that
	$W_i \in \restr{\TDown(U_{\rho(i)})}{H_{\rho(i)}}$.
	In practice, this amounts to the existence
	of an open $T_{\rho(i)}$,
	such that
	$T_{\rho(i)} \osub U_{\rho(i)}$,
	$T_{\rho(i)} \subseteq U_{\rho(i)}$, and
	$W_i = T_{\rho(i)} \setminus V_{\rho(i)}$.
	Without loss of generality we assume
	that $\rho$ is monotonic.

	Let us build
	the sequence $\mathcal Y$
	defined by
	$Y_i \defined U_i$ if $i < \rho(0)$
	and $Y_i \defined T_{\rho(i)}$ otherwise.
	This is a sequence of open sets in
	$\AExp^\alpha(\tau)$ that is lexicographically smaller
	than $\AUseq$, hence $\mathcal Y$ is a \kl(noeth){good sequence}:
	there exists $i \in \mathbb{N}$ such that
	$Y_i \subseteq \bigcup_{j < i} Y_j$.
	\begin{itemize}
		\item If $i < \rho(0)$,
		      then
		      $U_i \subseteq \bigcup_{j < i} U_j$
		      contradicting that $\AUseq$ is $\kl(noeth){bad}$.
		\item If $i \geq \rho(0)$,
		      let us write
		      $Y_i = T_{\rho(i)} \subseteq \bigcup_{j < \rho(0)} U_j
			      \cup \bigcup_{j < i} T_{\rho(j)}$.
		      By taking  the intersection with $H_{\rho(i)}$,
		      we obtain
		      $W_i \subseteq \bigcup_{j < i} W_j$,
		      contradicting the fact that $\mathcal W$
		      is a \kl(noeth){bad sequence}.
		      \qedhere
	\end{itemize}
\end{proof}

We are now ready to leverage our knowledge of
\kl(noeth){minimal topologies} associated with \kl(noeth){minimal bad sequences}
to carry on the proof of our main theorem.

\limittopologies*
\begin{proof}
	If $\alpha$ is a sucessor ordinal,
	then $\alpha = \beta + 1$ and $\AExp^{\alpha} (\tau) = \AExp(\AExp^\beta(\tau))$.
	Because $\AExp$ respects \kl{Noetherian} topologies, we immediately conclude
	that $\AExp^{\alpha}(\tau)$ is \kl{Noetherian}. We are therefore
	only interested in the case where $\alpha$ is a limit ordinal.

	Assume by contradiction that $\AExp^\alpha(\tau)$ is not
	\kl{Noetherian},
	using \cref{lem:gen:minbadseq}
	there exists a \kl(noeth){minimal bad sequence}
	$\AUseq \defined (U_i)_{i \in \mathbb{N}}$.
	Let us write $d_i \defined \depth(U_i) < \alpha$.
	Thanks to \cref{lem:gen:depthminbad},
	$d_i$ is either $0$ or a successor ordinal.

	Because $\AExp^{\beta}(\tau)$ is \kl{Noetherian}
	for $\beta < \alpha$,
	there are finitely many opens $U_i$ at depth $\beta$
	for every ordinal $\beta < \alpha$.
	Indeed, if they were infinitely many, one would
	extract an infinite bad sequence of opens
	in $\AExp^{\beta}(\tau)$, which is absurd.

	Furthermore, the sequence $(d_i)_{i \in \Nat}$
	must be monotonic, otherwise $\AUseq$
	would not be lexicographically minimal.
	We can therefore construct a strictly increasing map
	$\rho \colon \mathbb{N} \to \mathbb N$
	such that
	$0 < \depth(U_{\rho(j)})$
	and
	$\depth(U_{i}) < \depth(U_{\rho(j)})$
	whenever $0 \leq i < \rho(j)$.

	Let us consider some $i = \rho(n)$ for some $n \in \mathbb{N}$.
	Let us write $V_{i}
		\defined \bigcup_{j < i} U_j$, and
	$H_{i} \defined X \setminus V_i$.
	The set $V_i$ is open in $\TDown(U_i)$
	by construction of $\rho$,
	hence $H_i$ is closed in
	$\TDown(U_i)$.
	As $\AExp$ is a \kl{topology expander},
	we derive the following inclusions:
	\begin{align*}
		\restr{\AExp(\TDown(U_i))}{H_i} & \subseteq
		\restr{\AExp(\restr{\TDown(U_i)}{H_i})}{H_i} \\
		                                & \subseteq
		\restr{\AExp(\UTopo(\AExp^\alpha(\tau)))}{H_i}
	\end{align*}

	Recall
	that $U_{i} \in \AExp(\TDown(U_{i}))$
	thanks to \cref{lem:gen:tdown-u}.
	As a consequence,
	$U_i \setminus V_i = W_i \setminus V_i$ for some open set $W_i$
	in
	$\AExp(\UTopo(\AExp^\alpha(\tau)))$.
	Thanks to \cref{lem:gen:mintopo},
	and preservation of \kl{Noetherian topologies}
	through \kl{topology expanders},
	the latter is a Noetherian topology. Therefore,
	$(W_{\rho(i)})_{i \in \mathbb{N}}$ is a \kl(noeth){good sequence}.
	This provides an $i \in \mathbb N$
	such that
	$W_{\rho(i)} \subseteq \bigcup_{\rho(j) < \rho(i)} W_{\rho(j)}$.
	In particular,
	\begin{align*}
		U_{\rho(i)} \setminus V_{\rho(i)}
		 & =
		W_{\rho(i)}
		\setminus V_{\rho(i)}
		\subseteq
		\bigcup_{\rho(j) < \rho(i)} W_{\rho(j)} \setminus V_{\rho(i)}
		\subseteq
		\bigcup_{\rho(j) < \rho(i)} W_{\rho(j)} \setminus V_{\rho(j)}
		\\
		 & \subseteq
		\bigcup_{\rho(j) < \rho(i)} U_{\rho(j)} \setminus V_{\rho(j)}
		\subseteq
		\bigcup_{j < \rho(i)} U_{j}
		= V_{\rho(i)}
	\end{align*}
	This proves that
	$U_{\rho(i)} \subseteq V_{\rho(i)}$,
	i.e. that $U_{\rho(i)} \subseteq \bigcup_{j < \rho(i)} U_j$.
	Finally, this contradicts the fact that $\AUseq$ is \kl(noeth){bad}.
	\qedhere
\end{proof}

We have effectively proven that being well-behaved with respect
to closed subspaces is enough to consider least fixed points
of \kl{refinement functions}.
This behaviour should become clearer
in the upcoming sections, where we illustrate
how this property can be ensured both in the case
of \kl{Noetherian spaces} and \kl{well-quasi-orderings}.

\section{Applications of Topology Expanders}
\label{sec:transfinite}
We now briefly explore topologies that can be proven to be \kl{Noetherian}
using \cref{cor:gen:noethpres}.
It should not be surprising that both the topological Higman lemma
and the topological Kruskal theorem fit in the framework
of \kl{topology expanders}, as both were already proven using
a minimal bad sequence argument.
However, we will proceed to extend the use of \kl{topology expander}
to spaces for which the original proof did not use a minimal bad sequence
argument, and illustrate how they can easily be used to define new Noetherian
topologies.

\subsection{Finite words and finite trees}
As a first example,
we can easily recover the
\emph{topological Higman lemma}~\cite[Theorem 9.7.33]{goubault2013non}
because the \kl{subword  topology} is
the least fixed point of $\RegSubExp$ (\cref{lem:gen:regsubexp-lfp}),
which is a
\kl{topology expander} (\cref{lem:gen:regsubexp-topo-expander}).

It does not require much effort to generalise this proof scheme
to the case of the
\emph{topological Kruskal theorem}~\cite[Theorem 9.7.46]{goubault2013non}.
As a shorthand notation, let us write
$t \in \diamond U \langle V \rangle$
whenever there exists a subtree $t'$ of $t$
whose root is labelled by an element of $U$ and
whose list of children belongs to $V$.

\begin{definition}[{\cite[Definition 9.7.39]{goubault2013non}}]
	Let $(\Sigma,\theta)$ be a topological space. The space $\trees{\Sigma}$
	of finite
	trees over $\Sigma$ can be endowed with the \intro{tree topology},
	the coarsest topology such that
	$\diamond U \langle V\rangle$ is open whenever
	$U$ is an open set of $\Sigma$, and $V$ is an open set of
	$\words{\trees{\Sigma}}$ in its \kl{subword  topology}.
\end{definition}
\begin{definition}
	Let $(\Sigma, \theta)$ be a topological space.
	Let $\intro*\TreeExp$ be the function
	that maps a topology $\tau$ to the topology generated by the sets
	$\upset{\KruLeq}{U \langle V \rangle}$, for $U$ open in $\theta$,
	$V$ open in $\words{\trees{\Sigma}}$ with the \kl{subword  topology}
	of $\tau$.
\end{definition}
\vspace{-2ex}
\begin{restatable}{lemma}{fintreescoarsest}
	\label{lem:ft:tree-coarsest}
	The \kl{tree topology} is the least fixed point of $\TreeExp$,
	which is a \kl{topology expander}.
\end{restatable}
\begin{sendappendix}
	\ifsubmission
		\fintreescoarsest*
	\fi
	\begin{proof}
		The proof is follows the same pattern as for the \kl{subword  topology}.
		The only technical part is to notice that a
		downwards closed set $H$
		for $\KruLeq$
		satisfies $(\uparrow_{\KruLeq} U \langle V \rangle) \cap H
			= (\uparrow_{\KruLeq} U \langle \openword{V_1 \cap H, \dots, V_n
				\cap H}
			\rangle) \cap
			H$, whenever $V = \openword{V_1, \dots, V_n}$.
	\end{proof}
\end{sendappendix}
\begin{corollary}
	The \kl{tree topology} is \kl{Noetherian}.
\end{corollary}

\subsection{Ordinal words}
Let us now demonstrate how \cref{cor:gen:noethpres} can be applied
over spaces for which the original proof of Noetheriannes
did not use a minimal bad sequence argument.
For that, let us consider $\Sigma^{< \alpha}$ the set of words of ordinal length
less than $\alpha$, where $\alpha$ is a fixed ordinal. Since $\HigLeq$
is in general
not a \kl{wqo} on $\Sigma^{<\alpha}$ when $\leq$ is \kl{wqo} on $\Sigma$,
this also provides an example of a topological minimal bad sequence argument
that has no counterpart in the realm of \kl{wqos}.




\begin{definition}[{\cite{goubaultlarrecq2022infinitary}}]
	\label{def:ft:ordinal-words}
	Let $(\Sigma, \theta)$ be a topological space.
	The \intro{ordinal subword topology}
	over $\Sigma^{<\alpha}$ is the
	topology generated by the closed sets
	$F_1^{< \beta_1} \cdots F_n^{<\beta_n}$,
	for $n \in \mathbb{N}$,
	for $F_i$ closed in $\theta$,
	and where $F^{<\beta}$
	is the set of words of length less than $\beta$
	with all of their letters in $F$.
\end{definition}

The \kl{ordinal subword topology}
is \kl{Noetherian} \cite{goubaultlarrecq2022infinitary}, but the proof
is quite technical and relies on the in-depth study
of the possible inclusions between the subbasic closed sets.
Before defining a suitable \kl{topology expander},
given an ordinal $\beta$ and a set $U \subseteq \Sigma^{<\alpha}$, let us write
$w \in \beta \triangleright U$ if and only if
$w_{> \gamma} \in U$ for all $0 \leq \gamma < \beta$.

\begin{definition}
	Let $(\Sigma, \theta)$ be a topological space, and $\alpha$
	be an ordinal.
	The function
	$\intro*\RegSubExpOrd$ maps a topology $\tau$
	to the topology generated by the following sets:
	\begin{inparaitem}[]
		\item $\upset{\HigLeq}{UV}$ for $U,V$ opens in $\tau$;
		\item $\upset{\HigLeq}{\beta \triangleright U}$, for $U$ open in
		$\tau$, $\beta \leq \alpha$; 		\item $\upset{\HigLeq}{W}$, for $W$ open in $\theta$.
	\end{inparaitem}
\end{definition}
\vspace{-1ex}
\begin{restatable}{lemma}{regsubexpordinal}
	\label{lem:ft:regordsubword}
	Given a \kl{Noetherian space} $(\Sigma,\theta)$, and an ordinal $\alpha$.
	The map $\RegSubExpOrd$ is a \kl{topology expander},
	whose least fixed point contains the \kl{ordinal subword topology}.
\end{restatable}
\begin{sendappendix}
	\ifsubmission \regsubexpordinal* \fi
	\begin{proof}
		It is obvious that $\RegSubExpOrd$ is monotone. Moreover,
		the closed sets $H$ in $\RegSubExpOrd(\tau)$ are downwards closed
		with respect to $\HigLeq$. As a consequence,
		$(\upset{\HigLeq}{UV}) \cap H
			=
			(\upset{\HigLeq}{(U \cap H) (V\cap H)}) \cap H$,
		$(\upset{\HigLeq}{W}) \cap H
			=
			(\upset{\HigLeq}{(W \cap H)}) \cap H$,
		and
		$(\upset{\HigLeq}{ \beta \triangleright U}) \cap H
			=
			(\upset{\HigLeq}{\beta \triangleright (U \cap H)}) \cap H$.
		Hence, $\RegSubExpOrd$ respects subsets.
		To conclude that $\RegSubExpOrd$ is a \kl{topology expander},
		it remains to prove that it preserves \kl{Noetherian topologies}.

		\begin{claim}
			Let $\tau$ be a Noetherian topology. Then $\RegSubExpOrd(\tau)$
			is Noetherian.
		\end{claim}
		\begin{proof}
			As a consequence of
			\cref{lem:gen:reg-exp-refinement},
			the topology generated by the sets
			$\upset{\HigLeq}{UV}$,
			and $\upset{\HigLeq}{W}$ is Noetherian.
			Therefore, it suffices to check that the topology
			generated by the sets $\upset{\HigLeq}{\beta \triangleright U}$
			is Noetherian to conclude that $\RegSubExpOrd(\tau)$ is
			too.

			For that, consider a bad sequence $\beta_i \triangleright U_i$
			of open sets, indexed by $\mathbb N$.
			Because for all $i$, $\beta_i < \alpha + 1$,
			we can extract our sequence so that $\beta_i \leq \beta_{j}$
			when $i \leq j$. The extracted sequence is still bad.
			Because $\tau$ is Noetherian, there exists $i \in \mathbb N$
			such that $U_i \subseteq \bigcup_{j < i} U_j$.
			Let us now conclude that
			$\beta_i \triangleright U_i
				\subseteq \bigcup_{j < i} \beta_j \triangleright U_j$, which
			is in contradiction with the fact that the sequence is bad.

			Let $w \in \beta_i \triangleright U_i$, and assume by contradiction
			that for all $j < i$, there exists a $\gamma_j < \beta_j \leq
				\beta_i$
			such that $w_{> \gamma_j} \not \in U_j$.
			Let $\gamma \defined \max_{j < i} \gamma_j < \beta_i$.
			The word $w_{> \gamma}$ does not belong to $U_j$ for $j < i$,
			because $U_j$ is upwards closed for $\HigLeq$.
			As a consequence, $w_{> \gamma} \not \in \bigcup_{j < i}
				U_j$.
			However, $w_{> \gamma} \in U_i$, which is absurd.
			\claimqedhere
		\end{proof}

		We now have to check that every open set in the
		\kl{ordinal subword topology} is open in the least fixed point
		of $\RegSubExpOrd$. We prove by induction over $n$
		that a product $F_1^{<\beta_1} \dots F_n^{< \beta_n}$
		has a complement that is open.

		\begin{description}
			\item[Empty product] this is the whole space.
			\item[$P \defined F^{<\beta} P'$]
			      By induction hypothesis, ${P'}^c$ is an open $U$ in the least
			      fixed point topology.
			      Let us prove that $P^c = A \cup B$,
			      where $A \defined {\upset{\HigLeq}{\setof{ av }{
								      u \not \in F
								      \wedge
								      av \in U
							      }
					      }}$, and
			      $B \defined {\upset{\HigLeq}{(\beta \triangleright
						      U)}}$.

			      \begin{claim}
				      $P^c \subseteq A \cup B$.
			      \end{claim}
			      \begin{claimproof}
				      Let $w \not \in P$ and distinguish two cases.
				      \begin{itemize}
					      \item Either there exists a smallest $\gamma < \beta$
					            such that $w_\gamma \not \in F$.
					            In which case $w = w_{<\gamma} w_{\gamma} w_{>
								            \gamma}$. Since $\gamma < \beta$,
					            $w_{\leq \gamma} \in F^{<\beta}$,
					            hence $w_{> \gamma} \in U$
					            because $w \not \in P$.
					            As a consequence, $w \in A$.
					      \item Or $w_{\gamma} \in F$ for every $\gamma < \beta$.
					            However, this proves that $w_{> \gamma} \in U$
					            for every $\gamma < \beta$, which
					            means that $w \in B$. \claimqedhere
				      \end{itemize}
			      \end{claimproof}

			      \begin{claim}
				      $A \subseteq P^c$.
			      \end{claim}
			      \begin{claimproof}
				      Because $P$ is downwards closed for $\HigLeq$, it suffices to
				      check that every word $av$ with $a \not \in F$ and $av \in U$
				      lies in $P^c$.

				      Assume by contradiction that $av \in P$, then $av = u_1 u_2$
				      with $u_1 \in F^{<\beta}$ and $u_2 \in P'$.
				      Because $a \not \in F$, this proves that $u_1$ is the empty
				      word,
				      and that $u_2 = w \in P'$. This is absurd because
				      $w \in U = (P')^c$.
				      \claimqedhere
			      \end{claimproof}
			      \begin{claim}
				      $B \subseteq P^c$.
			      \end{claim}
			      \begin{claimproof}
				      Because $P$ is downwards closed for $\HigLeq$ it
				      suffices to check that every word $w \in \beta \triangleright
					      U$ lies in $P^c$.

				      Assume by contradiction that such a word $w$
				      is in $P$. One can write $w = uv$
				      with $u \in F^{<\beta}$ and $v \in P'$.
				      However, $|u| = \gamma < \beta$, and $\gamma + 1 < \beta$
				      because $\beta$ is a limit ordinal.
				      Therefore, $v = w_{> \gamma} \in U = (P')^c$
				      which is absurd.
				      \claimqedhere
			      \end{claimproof}

			      \begin{claim}
				      $A$ and $B$ are open in the least fixed point
				      of $\RegSubExpOrd$.
			      \end{claim}
			      \begin{claimproof}
				      The set $B$ is open because $U$ is open.
				      Let us prove by induction that whenever $U$
				      is open and $F$ is closed in $\theta$,
				      the set $F \rtimes U$ defined as
				      $\upset{\HigLeq} {\setof{av}{
							      a \not \in F, av \in U
						      }}
				      $ is open.
				      It is easy to check that
				      $F \rtimes ({\uparrow_{\HigLeq}} W)
					      = {\uparrow_{\HigLeq}} (W \cap F^c)
					      \cup
					      {\uparrow_{\HigLeq}} F^c W$.
				      Moreover,
				      $F \rtimes ({\uparrow_{\HigLeq}} UV)
					      = {\uparrow_{\HigLeq}} (F \rtimes U) V$.
				      Finally, for $\beta \geq 1$,
				      $F \rtimes ({\uparrow_{\HigLeq}} \beta \triangleright U)
					      =
					      {\uparrow_{\HigLeq}} F^c (\beta'
					      \triangleright U)$ with $\beta' = \beta$ if $\beta$
				      is limit, and $\beta' = \gamma$ if $\beta = \gamma + 1$.
				      \claimqedhere
			      \end{claimproof}

			      We have proven that $P^c$ is open.
			      \qedhere
		\end{description}
	\end{proof}
\end{sendappendix}
\begin{corollary}
	The \kl{ordinal subword topology} is Noetherian.
\end{corollary}


\subsection{Ordinal branching trees}

\AP As an example of a new \kl{Noetherian topology} derived using
\cref{cor:gen:noethpres}, we will consider
\intro{$\alpha$-branching trees}
$\intro*\omtrees{\Sigma}{\alpha}$, i.e.,
the least fixed point of the constructor
$X \mapsto \singleton + \Sigma \times X^{< \alpha}$
where $\alpha$ is a given ordinal.

\begin{definition}
	\label{def:ft:alpha-trees}
	Let $(\Sigma, \theta)$ be a \kl{Noetherian space}.
	The \intro{ordinal tree topology} over
	\kl{$\alpha$-branching trees} is the least fixed point
	of $\intro*\TreeExpOm$, mapping a topology $\tau$
	to the topology generated by the sets
	$\upset{\KruLeq}{U\langle V \rangle}$,
	where $U \in \theta$, $V$ is open in $(\omtrees{\Sigma}{\alpha})^{<\alpha}$
	with the \kl{ordinal subword topology}, and
	$U \langle V \rangle$ is the set of trees whose root is labelled
	by an element of $U$ and list of children belongs to $V$.
\end{definition}
\begin{theorem}
	\label{thm:ft:alpha-trees}
	The \kl{$\alpha$-branching trees} endowed with the
	\kl{ordinal tree topology} forms a \kl{Noetherian space}.
\end{theorem}
\begin{proof}
	It suffices to prove that $\TreeExpOm$ is a \kl{topology expander}.
	It is clear that $\TreeExpOm$ is monotone, and
	a closed set of $\TreeExpOm(\tau)$ is always downwards closed
	for $\KruLeq$. As a consequence, if $\tau \subseteq \TreeExpOm(\tau)$
	and $H$ is closed in $\tau$,
	$t \in V \defined (\upset{\KruLeq}{U\langle V \rangle}) \cap H$
	if and only if
	$t \in H$ and every children of $t$ belongs to $H$.
	Therefore,
	$(\upset{\KruLeq}{U\langle V \rangle}) \cap H
		= (\upset{\KruLeq}{U\langle V \cap H^{<\alpha} \rangle}) \cap H$.
	Notice that $H^{<\alpha} \cap V$ is an open of the
	\kl{ordinal subword topology} over $\restr{\tau}{H}$. As a
	consequence,
	$V \cap H \in \restr{\TreeExpOm(\restr{\tau}{H})}{H}$.

	Let us now check that $\TreeExpOm$ presreves \kl{Noetherian topologies}.
	Let $W_i \defined \upset{\KruLeq}{U_i \langle V_i \rangle}$
	be a $\mathbb{N}$-indexed sequence of open sets in $\TreeExpOm(\tau)$ where $\tau$
	is Noetherian.
	The product of the
	topology $\theta$ and the \kl{ordinal subword topology} over $\tau$
	is Noetherian thanks to
	\cref{claim:join-prod-noeth,lem:ft:regordsubword}.
	Hence, there exists a $i \in \mathbb{N}$ such that
	$U_i \times V_i \subseteq \bigcup_{j < i} U_j \times V_j$.
	As a consequence, $W_i \subseteq \bigcup_{j < i} W_j$.
	We have proven that $\TreeExpOm(\tau)$ is Noetherian.
	\qedhere
\end{proof}

At this point, we have proven that the framework of
\kl{topology expanders} allows to build non-trivial \kl{Noetherian} spaces.
We argue that
this bears several advantages over ad-hoc proofs:
\begin{inparaenum}[(i)]
	\item the ad-hoc proofs are often tedious and error prone
	\cite{goubault2013non,goubaultlarrecq2021infinitary,goubaultlarrecq2022infinitary}
	\item the verification that $\AExp$ is a topology expander on the other hand
	is quite simple
	\item the framework provides a good reason for which the desired topology
	is a sensible choice.
\end{inparaenum}
However, this setting is not quite satisfactory yet, as we do not provide
an automatic definition of the topology expander in the case of
an inductively defined space.

\section{Consequences on inductive definitions}
\label{sec:inductive}
So far, the process of constructing \kl{Noetherian spaces}
has been the following: first build a set of points,
then compute a topology that is Noetherian as a least fixed point.
In the case where the set of points
itself is inductively defined (such as finite words or finite trees),
the second step might seem redundant.

While inductive definitions are quite clear in the set theoretic
interpretation, we are interested in quasi-orderings
and topologies, for which the notion of least fixed-point
has to be precised. To that purpose, let us now introduce
some basic notions of category theory.

\AP In this paper only three categories will appear,
the \kl{category} $\intro*\CatSet$ of sets and functions,
the \kl{category} $\intro*\CatTop$ of topological spaces and continuous maps,
and
the \kl{category} $\intro*\CatOrd$ of quasi-ordered spaces and monotone maps.
Using this language, a unary constructor
$G$ in the \kl{algebra of wqos}
defines an \intro{endofunctor} from objects of the category $\CatOrd$
to objects of the category $\CatOrd$
preserving \kl{well-quasi-orderings}.

In our study of \kl{Noetherian spaces} (resp. \kl{well-quasi-orderings}),
we will often see constructors $G'$ as first building a new
set of structures, and then adapting the topology (resp. ordering)
to this new set.
In categorical terms, we are interested in \kl{endofunctors} $G'$
that are \kl{U-lifts} of \kl{endofunctors} on $\CatSet$.

\begin{sendappendix}
	\begin{definition}
		An \kl{endofunctor} $G'$
		of $\CatTop$ is a \intro{lift}
		of an \kl{endofunctor} $G$  of $\CatSet$
		if
		the following diagram commutes,
		where $U$ is the forgetful functor
		\begin{center}
			\begin{tikzcd}
				\CatTop \arrow[r, "G'"]
				\arrow[d, "U"]
				& \CatTop
				\arrow[d, "U"]
				\\
				\CatSet \arrow[r, "G"] & \CatSet
			\end{tikzcd}
		\end{center}
	\end{definition}
\end{sendappendix}

\subsection{Divisibility Topologies over Analytic Functors}
\AP

As noticed by \citet*{HASEGAWA2002113} and \citet*{freund2020kruskal},
usual orderings on words and trees can be derived from their
least fixed point definitions. We will provide a similar construction
for topological spaces.
However, we will avoid as much as possible the
use of complex machinery related to \kl{analytic functors},
and use as a definition an equivalent characterisation
given by
\citet[Theorem 1.6]{HASEGAWA2002113}.
For an introduction to analytic functors and combinatorial
species, we redirect the reader to~\citet*{10.1007/BFb0072514}.

\begin{definition}
	Given $\SCons$ an \kl{endofunctor} of $\CatSet$,
	the \intro{category of elements}
	$\CatEl(\SCons)$ has as objects pairs $(E,a)$ with $a \in \SCons(E)$,
	and as morphisms between $(E,a)$ and $(E',a')$
	maps $f \colon E \to E'$
	such that $\SCons_f (a) = a'$.
\end{definition}

As an intuition to the unfamiliar reader,
an element $(E,a)$ in $\CatEl(\SCons)$
is a witness that $a$ can be produced
through $\SCons$
by using elements of $E$. Morphisms of
elements are witnessing how relations
between elements of $\SCons(E)$ and $\SCons(E')$ arise
from relations between $E$ and $E'$.
As a way to define a ``smallest'' set of elements $E$
such that $a$ can be found in $\SCons(E)$, we rely on
transitive objects.

\begin{definition}
	A \intro{transitive object} in a category $\ACat$
	is an object
	$X$
	satisfying the following two conditions for every
	object $A$ of $\ACat$:
	\begin{inparaenum}[(a)]
		\item $\Hom(X, A)$ is non-empty;
		\item  The right action of $\Aut(X)$ on $\Hom(X, A)$
		by composition is transitive.
	\end{inparaenum}
\end{definition}

Given an object $A$ in a category $\ACat$,
one can build the \kl{slice category}
$\slice{\ACat}{A}$ whose objects are elements of $\Hom(B,A)$ when
$B$ ranges over objects of $\ACat$ and
morphisms between $c_1 \in \Hom(B_1,A)$ and $c_2 \in \Hom(B_2, A)$
are maps $f \colon B_1 \to B_2$ such that $c_2 \circ f = c_1$.
This notion of \kl{slice category} can be combined
with the one of \kl{transitive object} to build
so-called ``weak normal forms''.

\begin{definition}
	A \intro{weak normal form} of an object $A$ in a category $\ACat$
	is a \kl{transitive object} in $\slice{\ACat}{A}$.
\end{definition}

A category $\ACat$ has the \intro*\kl{weak normal form property}
whenever every object $A$ has a \kl{weak normal form}.
We are now ready to formulate a definition of
\kl{analytic functors} through the existence of
\kl{weak normal forms} for objects
in their \kl{category of elements}.

\begin{definition}
	An \kl{endofunctor} $\SCons$ of $\CatSet$
	is an \intro{analytic functor} whenever
	its
	category of elements
	$\CatEl(\SCons)$
	has the weak normal form property.
	Moreover; $X$ is a finite set for every weak normal form $f \in \Hom((X, x),
		(Y,y))$
	in $\slice{\CatEl(\SCons)}{(Y,y)}$.
\end{definition}

\begin{example}
	The \kl{functor} mapping $X$ to $\words{X}$
	is \kl{analytic}, and the weak normal form of a word $(\words{X}, w)$ is
	$(\letters(w), w)$ together with the canonical injection from $\letters(w)$ to
	$X$.  In this specific case, the \kl{weak normal forms} are in fact \kl{initial
		objects}.
\end{example}
\begin{example}
	The functor
	mapping $X$ to $X^{<\alpha}$ is not \kl{analytic} when $\alpha \geq \omega$,
	because of the restriction that weak normal forms are defined using
	finite sets.
\end{example}

Let us now explain how these \kl{weak normal forms} can be used to define a
\kl{support} associated to the \kl{analytic functor}. Given an \kl{analytic
	functor} $\AFunc$ and an \kl{element} $(X,x)$ in $\CatEl(\AFunc)$, there exists
a \kl{weak normal form} $f \in \Hom((Y,y),(X,x))$ in the \kl{slice category}
$\slice{\CatEl(\AFunc)}{(X,x)}$.  By definition, $f \colon Y \to X$ and
$\AFunc_f (y) = x$. We define $f(Y)$ as the \kl{support} of $x$ in $X$.

In turn,
this construction of support allows building a \kl{substructure ordering} on
initial algebras $(\mu \AFunc, \delta)$ of $\AFunc$: an element $a \in \mu
	\AFunc$ is a child of an element $b \in \mu \AFunc$ whenever $a = b$ or $a
	\in \supportW(\delta^{-1}(b))$. The transitive closure of the children relation
is called the \intro*\kl{substructure ordering} on $\mu \AFunc$, and
written $\sstructeq$.

\begin{example}
	The substructure ordering on $\InitSCons$
	for
	$\SCons(X) \defined \singleton + \Sigma \times X$
	is the suffix ordering of words.
\end{example}

As \kl{analytic functors} induce
a quasi-ordering on their \kl{initial algebras},
it is natural to import this quasi-ordering
when dealing with \kl{lifts} of \kl{analytic functors}
in the category $\CatOrd$. This follows
the construction of \citet*[Definition 2.7]{HASEGAWA2002113},
although this \kl{substructure ordering} is implicitly built.
Given a topology $\tau$ on $\mu\AFunc$,
one can build
open sets as ${\uparrow}_{\sstructeq} U$
for $U \in \tau$. Open sets of this new topology
are automatically upwards closed for $\sstructeq$.

\begin{definition}
	Let $\TCons \colon \CatTop \to \CatTop$
	be a lifting of an analytic functor $\SCons$, and $(\InitSCons, \delta)$
	an \kl{initial algebra} of $\SCons$. Moreover, we
	suppose that $\TCons$ preserves inclusions.
	The \intro*\kl{divisibility topology} over $\InitSCons$
	is the least fixed point of
	$\intro*\DivExp$, mapping
	$\tau$ to the topology generated by $\setof{
			{\uparrow}_{\sqsubseteq} \delta(U)
		}{U \text{ open in } \TCons(\InitSCons, \tau)}$.
\end{definition}

\vspace{-2ex}
\begin{restatable}{theorem}{thminducdivisibility}
	\label{thm:induc:divisibility}
	The \kl{divisibility topology}
	is \kl{Noetherian}.
\end{restatable}
\begin{sendappendix}
	\ifsubmission \thminducdivisibility* \fi
	\begin{proof}
		We prove that $\DivExp$ is a \kl{topology expander} and conclude thanks
		to \cref{cor:gen:noethpres}.
		\begin{enumerate}
			\item Let us prove that $\DivExp$ sends
			      \kl{Noetherian topologies} to \kl{Noetherian topologies}.
			      This is because it is the upwards closure of
			      the image of a \kl{Noetherian} topology through $\delta$.

			\item Let us show that $\DivExp$ is monotone.

			      Consider $\tau \subseteq \tau'$ two topologies on $\InitSCons$.
			      Let us write $X \defined (\InitSCons, \tau)$ and $Y \defined (\InitSCons,
				      \tau')$.
			      By definition of the inclusion of topologies,
			      there exists an \kl(topo){embedding} $\iota \colon X \to Y$ in $\CatTop$
			      whose underlying function is the identity on $\InitSCons$.
			      Because $\TCons$ preserves \kl(topo){embeddings},
			      $\TCons_{\iota}$ is an \kl(topo){embedding} from
			      $\TCons(X)$ to $\TCons(Y)$, that is,
			      an \kl(topo){embedding} from $(\TCons(\InitSCons), \TCons(\tau))$
			      to $(\TCons(\InitSCons), \TCons(\tau'))$.
			      Moreover, $U\TCons_{\iota} = \SCons_{U \iota}
				      = \SCons_{\Id_{\InitSCons}}
				      = \Id_{\InitSCons}$.
			      As a consequence, $\TCons(\tau) \subseteq \TCons(\tau')$
			      and $\DivExp(\tau) \subseteq \DivExp(\tau')$.

			\item Let us consider a \kl{Noetherian topology} $\tau$
			      such that $\tau \subseteq \DivExp(\tau)$,
			      $H$ closed in $\tau$,
			      and prove that $\restr{\DivExp(\tau)}{H}
				      \subseteq \restr{\DivExp(\restr{\tau}{H})}{H}$.
			      Because $\SCons$ is an \kl{analytic functor},
			      we can assume without loss of generality that
			      $\SCons(H) \subseteq \SCons(\InitSCons)$.

			      \begin{claim}
				      \label{fact:ind:closedinclusion}
				      $\delta^{-1}(H) \subseteq G(H)$
			      \end{claim}
			      \begin{claimproof}
				      Let $t \in H$, because $H$ is downwards closed for
				      $\sstructeq$, for every $u \in \supportW(\delta^{-1}(t))$,
				      $u \in H$.
				      As a consequence, $\supportW(\delta^{-1}(t)) \subseteq H$,
				      and this means that $\delta^{-1}(t) \in \SCons(H)$.
				      \claimqedhere
			      \end{claimproof}

			      Let $U = {\uparrow}_{\sqsubseteq} \delta(V)$
			      be an open set of $\DivExp(\tau)$. Notice that
			      $H$ is a closed subset of $\DivExp(\tau)$ because
			      $\tau \subseteq \DivExp(\tau)$.
			      Therefore,
			      \begin{align*}
				      U \cap H & = ({\uparrow}_{\sstructeq} \delta(V)) \cap H
				      = {\uparrow}_{\sstructeq}  (\delta(V) \cap H) \cap H
				      = {\uparrow}_{\sstructeq} (\delta(V) \cap
				      \delta(\SCons(H))) \cap H                               \\
				               & = {\uparrow}_{\sstructeq} \delta(V \cap
				      \SCons(H)) \cap H
			      \end{align*}
			      To conclude that $U \cap H$ is open in
			      $\restr{\DivExp(\restr{\tau}{H})}{H}$
			      it suffices to show that
			      $V \cap \SCons(H)$ can be rewritten as
			      $W \cap \SCons(H)$ where $W$ is open in $\TCons(\InitSCons, \restr{\tau}{H})$.
			      Let us consider
			      two maps $e_1 \colon (H,\tau_H) \to (\InitSCons, \tau)$,
			      and $e_2 \colon (H,\tau_H) \to (\InitSCons, \restr{\tau}{H})$.
			      These two maps are embeddings, hence preserved by
			      $\TCons$. As a consequence,
			      $V \cap \SCons(H) = (\TCons_{e_1})^{-1}(V)$, which is open.
			      Because $\TCons_{e_2}$ is an embedding,
			      there exists a $W$ open in $\TCons(\InitSCons, \restr{\tau}{H})$
			      such that $(\TCons_{e_2})^{-1}(W) = V \cap \SCons(H)$.
			      This can be rewritten, $W \cap \SCons(H) = V \cup \SCons(H)$.
			      \qedhere
		\end{enumerate}
	\end{proof}
\end{sendappendix}

As a sanity check, we can apply \cref{thm:induc:divisibility}
to the sets of finite words and finite trees, and recover
the \kl{subword  topology} and the \kl{tree topology}
that were obtained in an ad-hoc fashion in \cref{sec:transfinite}.
In addition to validating the usefulness of
\cref{thm:induc:divisibility}, we believe that these are strong indicators
that the topologies introduced prior to this work were the right generalisations
of \kl{Higman's word embedding} and \kl{Kruskal's tree embedding}
in a topological setting, and addresses the canonicity issue
of the aforementioned topologies.

\begin{restatable}{lemma}{lemindhigman}
	\label{lem:ind:higman}
	The \kl{subword  topology} over $\words{\Sigma}$,
	is the \kl{divisibility topology}
	associated to the analytic functor $X \mapsto \singleton + \Sigma \times X$.
\end{restatable}
\begin{sendappendix}
	\ifsubmission \lemindhigman* \fi
	\begin{proof}
		It suffices to remark that the functions $\DivExp$ and
		$\RegSubExp$ have the same least fixed point, and
		conclude using \cref{lem:gen:regsubexp-lfp}.
	\end{proof}
\end{sendappendix}
\vspace{-3ex}
\begin{restatable}{lemma}{lemindkruskal}
	\label{lem:ind:kruskal}
	The \kl{tree  topology} over $\trees{\Sigma}$,
	is the \kl{divisibility topology}
	associated to the analytic functor $X \mapsto X \times \Sigma^*$.
\end{restatable}
\begin{sendappendix}
	\ifsubmission \lemindkruskal* \fi
	\begin{proof}
		It suffices to remark that the functions $\DivExp$ and
		$\TreeExp$ have the same least fixed point, and
		conclude using \cref{lem:ft:tree-coarsest}.
	\end{proof}
\end{sendappendix}

\subsection{Divisibility Preorders}
\AP

We are now going to prove that the \kl{divisibility topology}
correctly generalises the corresponding notions on quasi-orderings.
In the case of finite words, this translates
to the equation $\Alex{\leq}^* = \Alex{\leq^*}$~\cite[Exercise 9.7.30]{goubault2013non}.
We will proceed to generalise this result to every
\kl{divisibility topology} by relating it to the
\kl{divisibility preorder} introduced by
\citet[Definition 2.7]{HASEGAWA2002113}.

Given an \kl{analytic functor} $\SCons$
and its lift $\OCons$ to quasi-orderings respecting
embeddings and \kl{wqos},
let us build a family $A_i$ of quasi-orders
and $e_i \colon A_i \to A_{i+1}$ of embeddings
as follows:
\begin{itemize}
	\item $A_0 = \emptyset$, $A_1 = \OCons(A_0)$
	      and $e_0$ is the empty map.
	\item $e_{n+1} = \OCons_{e_n}$ and
	      $A_{n+1}$ has as carrier set $\SCons(A_n)$
	      and preordering the transitive closure of the union
	      of the two following relations:
	      The one is the quasi-order $\OCons(A_n)$,
	      and the other is the collection of $b \dl a$
	      for each \kl{weak normal form}
	      $(X,z) \to^f (A_n, a)$ in $\CatEl(\SCons)$
	      and each $b$ in the image of
	      $X \to^f A_n \to^{e_n} A_{n+1}$.
\end{itemize}
The \intro*\kl{divisibility ordering} $\preceq$
is the $\omega$-inductive limit
in the category $\CatOrd$ of the diagram
$A_0 \to^{e_0} A_1 \to^{e_1} \cdots$.
As remarked by \citeauthor*{HASEGAWA2002113},
the maps $e_n$ are injective \kl{order embeddings}, and so are
the morphisms $c_n \colon A_n \to \InitSCons$ of the colimiting cone
\cite[Lemma 2.8]{HASEGAWA2002113}.
Without loss of generality, we can assume that $A_0 \subseteq A_1 \dots$
and that the colimit $\InitSCons$ is the union of the sets $A_i$ for $0 \leq i
	< \omega$. In particular, the map $\delta$ is the identity map in this
setting.

\begin{sendappendix}
	\begin{lemma}
		$a \dl b$ in $A_{n+1}$ if and only if $a \in \supportW(\delta^{-1}(b))$.
	\end{lemma}
	\begin{proof}
		Assume that $a \dl b$, then $b \in \SCons(A_n)$
		and there exists a \kl{weak normal form}
		$(X,z) \to^f (A_n, b)$ such that
		$a \in f(X)$. As $(A_n,b) \to^\iota (\InitSCons, b)$,
		$(X,z) \to^{f\iota} (\InitSCons, b)$
		is also a \kl{weak normal form}
		\cite[Lemma 1.5]{HASEGAWA2002113}.
		As a consequence, $a \in \iota(f(X))$ and
		$a \in \supportW(\delta^{-1}(b))$.

		Assume that $a \sstruct b$, there exists
		a \kl{weak normal form} $(X, z) \to^f (\InitSCons, b)$
		such that $a \in f(X)$. As $b \in \SCons(A_n)$ for some $n \in \mathbb{N}$,
		this means that $(A_n, b) \to^\iota (\InitSCons, b)$ is
		an element of the slice category, hence that
		there exists $g$ such that
		$(X,z) \to^g (A_n, b)$ is a \kl{weak normal form} of $b$
		and $\iota \circ g = f$. In particular, $a \in g(X)$,
		hence $a \in \supportW(\delta^{-1}(b))$.
	\end{proof}

	A direct consequence is that our substructure relation
	captures the height of the sets $A_n$ in the following sense:

	\begin{claim}
		\label{fact:ind:down-strict-an}
		If $a \sstruct b$ and $b \in A_{n+1}$ then $a \in A_n$.
	\end{claim}
	\begin{claim}
		\label{fact:ind:an-down-closed}
		For all $n \in \mathbb{N}$,
		$A_n$ is a downwards closed subset of $A_{n+1}$.
	\end{claim}

	Now, it is an easy check that the \kl{divisibility preorder}
	on $\InitSCons$ is compatible with substructures
	as this is true for the sets $A_n$.
\end{sendappendix}
\vspace{-1ex}
\begin{restatable}{lemma}{lemindpreorderstability}
	\label{lem:ind:preorder-stability}
	We have $(\preceq \sstructeq)^* = \preceq$.
\end{restatable}
\begin{sendappendix}
	\ifsubmission \lemindpreorderstability* \fi
	\begin{proof}
		By induction we prove it on $A_n$
		using the fact that $a \sstruct b$ and $b \in A_{n+1}$
		implies $a \in A_n$, thanks to
		\cref{fact:ind:down-strict-an,fact:ind:an-down-closed}.
	\end{proof}
\end{sendappendix}
\vspace{-2ex}
\begin{restatable}{corollary}{coralexdiprevdivtopo}
	The \kl{Alexandroff topology} of the \kl{divisibility preorder}
	contains the \kl{divisibility topology}.
\end{restatable}
\begin{sendappendix}
	\ifsubmission \coralexdiprevdivtopo* \fi
	\begin{proof}
		It suffices to prove that $\DivExp(\Alex{\preceq}) \subseteq
			\Alex{\preceq}$.
		Let us consider an open set $V$ of $\DivExp(\Alex{\preceq})$
		of the form $\uparrow_{\sstructeq} \delta(U)$, where
		$U$ is open in $\Alex{\preceq}$.
		In particular, $U = \uparrow_{\preceq} U$.
		Notice that $\uparrow_{\sstructeq} \uparrow_{\preceq} U = U$
		because of \cref{lem:ind:preorder-stability}. We have proven that
		$V \in \Alex{\preceq}$.
	\end{proof}
\end{sendappendix}
\vspace{-2ex}
\begin{sendappendix}
	\begin{lemma}
		For all $n \in \mathbb{N}$,
		\begin{equation*}
			(\sstructeq \leq_{\OCons(A_n)})^* \sstruct
			\quad = \quad
			\leq_{\OCons(A_n)} \sstruct
		\end{equation*}
		Note that this equality is only over elements of $A_{n+1}$.
	\end{lemma}
	\begin{proof}
		Let $n \in \mathbb{N}$.
		Only one inclusion is non trivial.
		We know that $\leq_{A_{n+1}} = (\leq_{\OCons(A_n)} \sstructeq)^*$.
		As the maps $e_n$ is an \kl{order embedding},
		for every $a, b \in A_n$,
		$a \leq_{A_{n+1}} b$ implies $a \leq_{A_n} b$. In particular,
		$\leq_{A_{n+1}} \sstruct = \leq_{A_n} \sstruct$.
		As $e_n$ is monotone from $A_n$ to $\OCons(A_n)$,
		$x \leq_{A_n} y$ implies $x \leq_{\OCons(A_n)} y$
		and therefore $(\leq_{A_{n+1}}\sstruct) \subseteq (\leq_{\OCons(A_n)} \sstruct)$.
	\end{proof}

	\begin{corollary}
		\label{cor:ind:fan-anplus1}
		For all $n \in \mathbb{N}$,
		$\leq_{\OCons(A_n)} \sstructeq = \leq_{A_{n+1}}$
	\end{corollary}
\end{sendappendix}

\begin{restatable}{lemma}{indpreorderanplusone}
	For all $n \in \mathbb{N}$,
	$\Alex{\preceq_n} \subseteq \DivExp(\Alex{\preceq_{n}})$,
	where $\preceq_n = \restr{\preceq}{A_n}$.
\end{restatable}
\begin{sendappendix}
	\ifsubmission \indpreorderanplusone* \fi
	\begin{proof}
		Let $x \in \InitSCons$ and
		consider $U = \uparrow_{\preceq_n} x$, which is open
		in $\Alex{\preceq_n}$.
		Let us write $V = \uparrow_{\preceq_{n+1}} \setof{y}{x \preceq_n y}$.
		It is clear that $U = V$, let us now prove that
		$V$ is open in $\DivExp(\Alex{\preceq_n})$.

		Thanks to \cref{cor:ind:fan-anplus1},
		$V = \uparrow_{\sstructeq} \uparrow_{F(\preceq_n)}
			\setof{y}{x \preceq_n y}$. Moreover,
		$\uparrow_{\OCons(\preceq_n)}
			\setof{y}{x \preceq_n y}$ is open in $\TCons(\InitSCons,
			\Alex{\preceq_n})$.
		As a consequence, we have proven that $V$ is open in
		$\DivExp(\Alex{\preceq_n})$.
	\end{proof}
\end{sendappendix}

\begin{corollary}
	$\Alex{\preceq}$ is contained in the \kl{divisibility topology}.
\end{corollary}

We are now ready to state our correctness theorem, i.e., that
the \kl{divisibility topology} is a correct generalisation to the
topological setting of the \kl{divisibility preorder}
from \citeauthor*{HASEGAWA2002113}.

\begin{theorem}
	\label{thm:induc:coincidence-alexandroff}
	Let $\TCons$ the be the lift of an \kl{analytic functor}
	respecting \kl{Alexandroff topologies},
	\kl{Noetherian spaces}, and \kl[topology embeddings]{embeddings}.
	Then, the \kl{divisibility topology} of $\InitSCons$
	is the \kl{Alexandroff topology}
	of the \kl{divisibility preorder} of $\InitSCons$,
	which is a \kl{well-quasi-ordering}.
\end{theorem}

\section{Concluding Remarks}
\label{sec:conclusion}
We have provided a systematic way to place a \kl{Noetherian topology}
over an inductively defined datatype, which is correct with respect
to its \kl{wqo} counterpart whenever it exists.
As a byproduct,
we obtained a uniform framework that simplifies existing
proofs, and serves as an indicator
that the pre-existing topologies were the ``right generalisations''
of their quasi-order counterparts. Let us now briefly
highlight some interesting properties of the underlying theory.

\subparagraph*{Differences with the existing categorical frameworks.}
The existing categorical frameworks are built around specific
kind of functors \cite{HASEGAWA2002113,freund2020kruskal},
while the notion of \kl{topology expander} only requires talking about
one specific set. This allows proving that the \kl{ordinal subword topology}
and the \kl{$\alpha$-branching trees} are \kl{Noetherian}, while
these escape both the realm of \kl{wqos}, and of ``well-behaved functors''
having finite support functions.

\subparagraph*{Quasi-analytic functors.}
In fact, the proof of
\cref{thm:induc:divisibility},
never relies on the finiteness of the support of an element.
This means that the definition of
\kl{analytic functors} can be loosened to allow non finite
weak normal forms. We do not know
whether this notion of ``quasi-analytic functor'' already
exists in the literature.

\subparagraph*{Transfinite iterations.}
As the reader might have noticed, all of the least fixed points
considered in this paper are obtained using at most $\omega$ steps.
This is because the \kl{topology expanders} that are presented in the paper
are all Scott-continuous, i.e., they satisfy the equation
$\AExp(\sup_i \TP_i) = \sup_i \AExp(\TP_i)$.
While \cref{cor:gen:noethpres} does apply to non Scott-continuous
topology expanders, we do not know any
reasonable example of such expander.

\subparagraph*{Lack of ordinal invariants.}
Even though our proof that the
\kl{ordinal subword topology} is \kl{Noetherian} is shorter
than the original one,
it actually provide less information.
In particular, it does not provide a bound for ordinal rank of the lattice of closed
sets (called the \intro{stature} of $\Sigma^{<\alpha}$), whereas
a clear bound is provided by the previous approach \citet[Proposition 33]{goubaultlarrecq2022infinitary}.
This limitation already appears in the existing categorical frameworks
\cite{HASEGAWA2002113,freund2020kruskal}, and
we believe that this is inherent to the use of minimal bad sequence
arguments.

\bibliography{globals/ressources}



\end{document}